\documentclass[journal,comsoc]{IEEEtran}
\usepackage{ulem}
\usepackage{amsmath,amssymb,amsfonts}
\usepackage{algorithmic}
\usepackage{graphicx}
\usepackage{textcomp, hyperref}
\usepackage{xcolor}
\usepackage{enumitem}
\usepackage{amsthm}
\usepackage{blindtext}
\usepackage{amsmath}
\usepackage{subfigure}
\usepackage{epstopdf}
\usepackage{wrapfig}
\newtheorem{lemma}{Lemma}

\newtheorem{theorem}{Theorem}
\newtheorem{definition}{Definition}
\newtheorem{corollary}{Corollary}
\newtheorem{assumption}{Assumption}

\newtheorem{remark}{Remark}

\ifCLASSOPTIONcompsoc
  \usepackage{cite}
\else
  \usepackage{cite}
\fi
\ifCLASSINFOpdf
\else
\fi
\begin{document}
%
\title{Optimal EV Charging Decisions Considering Charging Rate Characteristics and\\ Congestion Effects}
%
%
%
%

\author{Lihui~Yi
        and~Ermin~Wei
\IEEEcompsocitemizethanks{\IEEEcompsocthanksitem The authors are with Northwestern University, Evanston,
IL, 60208. L. Yi is with the Department
of Electrical and Computer Engineering (e-mail: lihuiyi2027@u.northwestern.edu).
 E. Wei is with the Department
of Electrical and Computer Engineering and Department of Industrial Engineering and Management Sciences (e-mail: ermin.wei@northwestern.edu).}
}

\IEEEtitleabstractindextext{%
\begin{abstract}
With the rapid growth in the demand for plug-in electric vehicles (EVs), the corresponding charging infrastructures are expanding. These charging stations are located at various places and with different congestion levels. EV drivers face an important decision in choosing which charging station to go to in order to reduce their overall time costs. {However, existing literature either assumes a flat charging rate and hence overlooks the physical characteristics of an EV battery where charging rate is typically reduced as the battery charges, or ignores the effect of other drivers on an EV's decision making process. In this paper, we consider both the predetermined exogenous wait cost and the endogenous congestion induced by other drivers' strategic decisions, and propose a differential equation based approach to find the optimal strategies. We analytically characterize the equilibrium strategies and find that co-located EVs may make different decisions depending on the charging rate and/or remaining battery levels. Through numerical experiments, we investigate the impact of charging rate characteristics, modeling parameters and the consideration of endogenous congestion levels on the optimal charging decisions. Finally, we conduct numerical studies on real-world data and find that some EV users with slower charging rates may benefit from the participation of fast-charging EVs.} 
\end{abstract}

}

\maketitle

\IEEEdisplaynontitleabstractindextext

%
\IEEEpeerreviewmaketitle

\section{Introduction}

%
%
%
%
\IEEEPARstart{T}{he} federal government aims to achieve a goal of making $50\%$ of new vehicles sold in the U.S. by 2030 be electric vehicles (EVs). To achieve this, they plan to establish a convenient and equitable network of 500,000 chargers to make EVs accessible for both local and long-distance trips \cite{2030 goal}. As EVs become more widely adopted, there has been significant research into designing efficient charging infrastructures. For example, in \cite{EV parking conference} and \cite{EV parking journal}, the authors study the impact of government policies on firms' investment in EV charging infrastructures. In \cite{battery swapping}, the authors design an online algorithm to assign service stations in response to EVs' battery swapping requests, based on the location of requesting EVs and the availability of fully-charged batteries at service stations in the system.
    
    On the EV side, every manufacturer builds a Battery Management System (BMS) in order to keep the battery healthy and increase its longevity. The BMS captures real time information of the battery (e.g., temperature, charging level), negotiates with the charger for proper voltage, current and charging rate \cite{BMS}. 
    In \cite{battery behavior}, researchers study battery behavior from the empirical charging data. Due to the physical attributes of a battery,  
    the charging rate slows down as the battery level becomes higher \cite{battery model}. 
    In particular, the U.S. Department of Transportation states that when the battery gets close to full, it can take about as long to charge the last 10 percent of an EV battery as the first 90 percent \cite{charging rate at 90 percent}.
    
    This work is related to the growing literature on control and optimization of EVs \cite{EV charging coordination, Ridesharing system, scheduling algorithm for ACN, eco-driving control, V2G}, especially those studying the placement of EV charging stations and  the interactions between charging stations and drivers \cite{optimization model, model with EV bus, model with self-interested EVs, model with various networks}. The work in \cite{optimization model} proposes an optimization model to formulate the EV charging station placement problem, analyzes its complexity and proposes several solution methods. In \cite{model with EV bus}, the authors consider the problem of placing EV charging stations at selected bus stops, to minimize the total installation cost. In \cite{model with self-interested EVs}, to capture the competitive and strategic charging behaviors of EV users, a model considering environment factors such as queuing conditions in charging stations is proposed. In \cite{model with various networks}, researchers propose a multi-stage charging station placement strategy that takes into account the interactions among EVs, the road network, and the electric power grid. 
    However, they often ignore EVs' varying charging rate as a function of their charging levels. Based on charging rate curves of commonly seen EVs, such as the Tesla Model 3, the charging speed at high battery levels is generally lower than that at low battery levels \cite{tesla model 3}. 
    \begin{wrapfigure}{l}{0.23\textwidth}
        \includegraphics[width=1.1\linewidth]{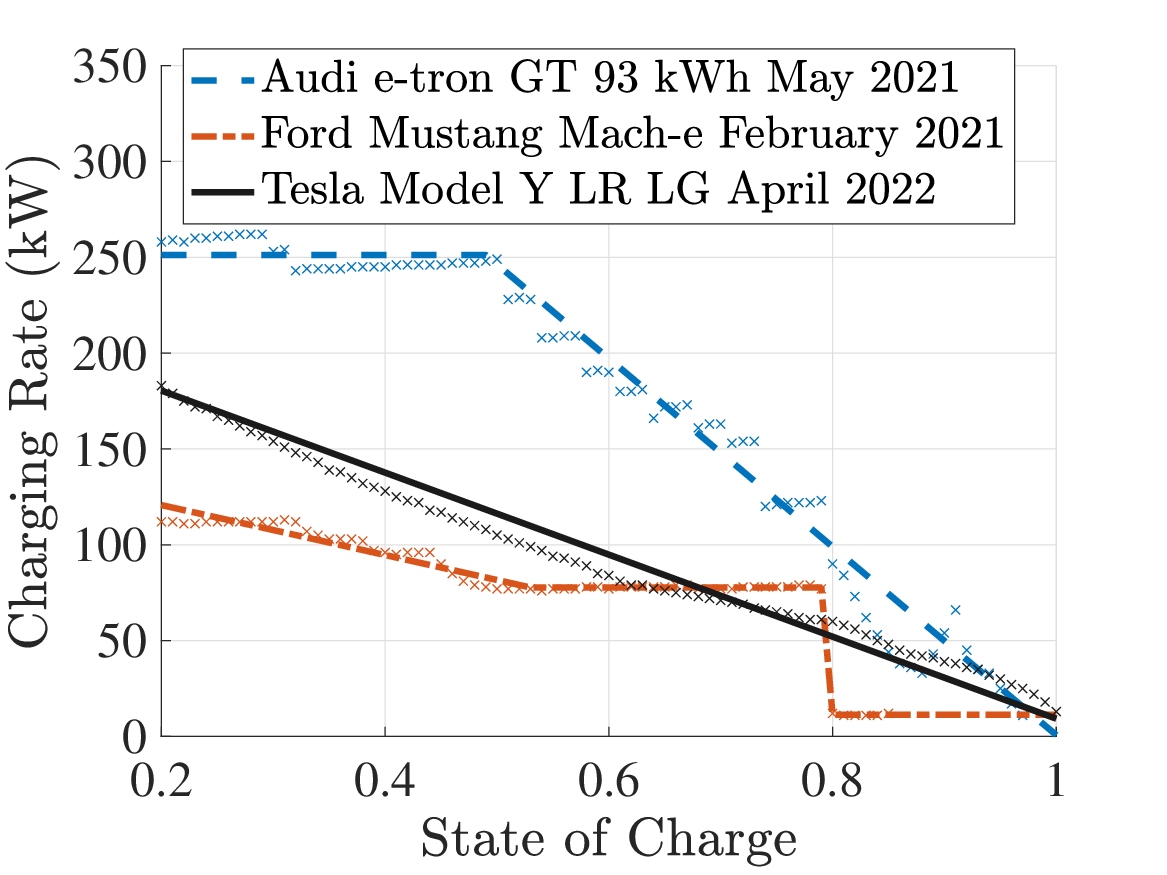}
      \caption{EV charging rate curves of 3 specific makes and models.}
      \label{field data charging rates}
    \end{wrapfigure}
    Field data of EV charging rates collected by Bjørn Nyland also shows such phenomenon \cite{field data}. In particular, we take the Audi e-tron GT, Ford Mustang Mach-e and Tesla Model Y LR LG as examples for demonstration. Fig. \ref{field data charging rates} shows the raw data points and the corresponding charging rate curves we derived using curve fitting techniques.

In this paper, we model and study the decision making process of EV drivers when choosing between two charging stations that are situated in different locations and have varying congestion levels. Our study takes into account the effects of varying charging rates of electric vehicles. Our stylized model is inspired by the classical Hotelling's model \cite{hotelling model}, where customers  choose between two firms based on spatial proximity and firms choose their locations accordingly. In our case, two charging stations are situated at two extreme points on a line segment, each with a waiting time determined exogenously and/or endogenously. EVs with different battery levels and locations must decide which charging station to use to minimize their overall time cost, which includes waiting time, traveling time from their current position to the charging station, and charging time. {Our prior work \cite{allerton} considers a special case, where there is no endogenous congestion induced by other drivers' strategic decisions and simply assumes a fixed and predetermined waiting time. The endogenous equilibrium here is challenging due to the dependency between the optimal strategies of one single EV driver and all the other drivers, thus the method in \cite{allerton} cannot be easily extended. To address this problem, we propose a differential equation based approach that involves studying an ordinary differential equation and connecting its solution to the optimal decision making of EV drivers.}

    Our main contribution is the introduction of charging time function, which captures the varying levels of charging rates and the study of their impact on selecting the optimal charging station. We demonstrate that optimal strategies for EVs, given varying charging rates, differ from those with constant rates. {Moreover, we model both exogenous and endogenous waiting time costs and analyze the corresponding equilibrium. Finally, numerical and field data experiments illustrate the optimal strategies for EV drivers in various scenarios.}

    {The remaining sections of this paper are organized as follows. In Section \ref{model}, we introduce our model. Section \ref{Optimal strategies} analyzes the optimal strategies for EV drivers under varying charging rates. Section \ref{sec: decreasing charging rates} considers the scenario when charging rates are non-increasing. In Section \ref{numerical exogenous model} and \ref{numerical endogenous model}, we present numerical case studies for exogenous model and endogenous model, respectively. In Section \ref{heterogeneous}, we discuss the heterogeneous case where various charging rate characteristics are considered together in one single game. Finally, Section \ref{conclusion} concludes the paper.}

\section{Model}\label{model}
    In this section, we present our model and define some key terms.
    
    \subsection{Charging Stations}
        We assume there are two EV charging stations, 
        identified by Station A and Station B. Suppose the location of Station A is $0$ and the location of Station B is normalized to be $1$, which are two extreme points of the space we are interested in. Both charging stations offer the same product, namely electricity, to customers located on a line segment of $[0, 1]$. 
        
    \subsection{EV Drivers}
        We assume there is a continuum  of EV drivers located on the line segment between two charging stations and each EV driver wishes to charge at one of them. The location of an EV is denoted by $y \in [0, 1]$. 
        
        Each EV may have a different remaining battery level $r \in [c, r_{t}]$, where $0 < c < r_{t} \le 1$. Here, constants $c$ and $r_{t}$ represent the power consumption to travel one unit distance and the target battery capacity level that all EVs wish to charge to, respectively. We assume the values of $c$ and $r_t$ are the same for all EVs. By letting EVs' battery levels be greater than or equal to $c$, all EV drivers can reach both charging stations regardless of their positions. 
        
        Each EV driver is characterized by the parameter pair $(r, y)$. We consider the EV drivers in the region $R$ corresponding to the product space of $[c, r_{t}]\times [0, 1]$ with $r \in [c, r_{t}]$ and $y \in [0, 1]$. These drivers need to decide simultaneously which charging stations (A or B) to charge at to minimize their {overall} time costs. 
        
        \subsection{Cost Functions} The overall time cost function is comprised of 3 components: waiting time, traveling time and charging time.

        {
        \subsubsection{Waiting time}
            Waiting time refers to the time cost of waiting at Station A or B. Here, we introduce two separate models for  waiting time: exogenous and endogenous.

            In the \textit{exogenous  model}, waiting times at Station A and B are denoted by two non-negative constants $w_{A}^{x}$ and $w_{B}^{x}$, respectively. This simple model was initially proposed for tractability purposes in our preliminary work \cite{allerton}, which assumes the waiting times to be independent of EV drivers' strategies. This assumption is reasonable when the drivers make their decisions without accounting for the other drivers. 

            In the \textit{endogenous model}, we consider a more realistic scenario that waiting times vary with the number of drivers who choose a particular charging station. That is to say, the strategies of EV drivers may cause congestion and affect the value of waiting times. 
            
            More formally, we model the congestion  as a function of the proportion of EVs that opt for a particular charging station. The more EVs choose Station A (or B), the less desirable it is on average for the other EVs. Since we assume all EV drivers make their one-shot decisions simultaneously, they all see the same congestion level. Define $\epsilon$ ($\epsilon > 0$) as the congestion level when all EV drivers opt for charging station A or B. We then follow the idea in \cite{linear congestion level} and model the congestion at Station A and B as two linear functions $\epsilon \alpha$ and $\epsilon (1-\alpha)$, respectively, where $\alpha$ ($0 \le \alpha \le 1$) is the proportion of EVs that choose Station A as their optimal strategy. 
            
            For the endogenous model, we add the congestion time to the exogenous waiting time  and define the total waiting time at Station A (B) as 
            \begin{align*}
                w_A^{n} = w_{A}^{x} + \epsilon\alpha,
            \qquad
                w_B^{n} = w_{B}^{x} + \epsilon(1-\alpha).
            \end{align*}
            }
            
        \subsubsection{Traveling time}
            Traveling time cost refers to the time taken to reach the selected charging station. We simplify the model by assuming all EVs travel at the same speed. We denote $\tau > 0$ as the unit distance traveling time or the total time taken to travel from Station A (B) to Station B (A). For an EV driver with the parameter pair $(r, y)$, 
            their traveling time to Station A (B) is
            \begin{align*}
                \tau y,
            \qquad
                \tau (1-y).
            \end{align*}
            
        \subsubsection{Charging time}
            {
            To model different battery sizes, we use $E$ ($E>0$) to denote the battery capacity of the EVs and $P: [0,1] \to \mathbb{R^+}$ as the charging rate which is a function of the state-of-charge (SoC), $r$. {Here, we state our assumption on $P$ below.}

            { \begin{assumption}
                We consider charging rate function $P$ to be continuous in $[0, 1]$.
            \end{assumption}

            }
            
            Next, we define the charging time $F(r)$ as the time  required to charge an EV from the current battery level $r$ to the target battery capacity level $r_{t}$, i.e.,
            \begin{align*}
                F(r) = \int_{r}^{r_t} \frac{E}{P(\gamma)}d\gamma.
            \end{align*}
           We observe that $F$ is a continuously differentiable and strictly decreasing function \cite{rudin}. Note that different make and model of EVs may have different charging rates and charging time functions. For simplicity, we first assume all EVs have the same charging characteristics. We will discuss the case of heterogeneous EVs with different $P$ and $F$ functions in Section \ref{heterogeneous}.
            }

            For an EV located at $y$ with remaining battery level $r$, its charging time at Station A is
            $ F(r-cy),$       
            where $r-cy$ represents the remaining battery level after driving to Station A. Similarly, the charging time at Station B is
            $F(r-c(1-y))$,
            where $r-c(1-y)$ is the remaining battery level upon arrival at Station B.
            
        \subsubsection{Total time cost}
            The total time cost is the sum of waiting time, traveling time and charging time. Specifically, the total time cost when EVs choose Station A and Station B is written as
            \begin{align*}
                &T_{A}^{i}(r, y) = w_{A}^{i} + \tau y + F(r-cy),\\
                &T_{B}^{i}(r, y) = w_{B}^{i} + \tau (1-y) + F(r-c(1-y)),
            \end{align*}
            where $i=x$ represents the exogenous model and $i=n$ represents the endogenous model.

            To understand how EV drivers make decisions, we define 
            $$\Delta T^{i}(r, y) = T_{B}^{i}(r, y) - T_{A}^{i}(r, y).$$
            Then, for EVs with parameter pair $(r,y)$, three possible cases may arise. If $\Delta T^{i}(r, y) > 0$, they should choose Station A. If $\Delta T^{i}(r, y) < 0$, they should choose Station B. If $\Delta T^{i}(r, y) = 0$, they are indifferent between the two charging stations. We refer to this last group as the {\it indifferent drivers}.

\section{Optimal Strategies for EV Drivers}\label{Optimal strategies}
    In this section, we study the optimal strategy for an EV driver associated with the parameter pair $(r, y)$. We address the problem by first characterizing the existence and properties of the set of indifferent drivers. Then, the optimal strategies for other EVs follow naturally. 
    
    The set of indifferent drivers is those within region $R$ that have equal costs to both stations.   Our preliminary work \cite{allerton}  establishes conditions for the existence of these drivers under specific assumptions.  In this work, we take an alternative approach, where we first extend the region from $R:[c, r_t]\times [0,1]$ to $(-\infty, \infty)\times(-\infty,\infty)$. To do so, we extend  the domain of function  $F$  from $[0, r_t]$  to $(-\infty, \infty)$, such that we preserve the continuously differentiable and strictly decreasing properties.  Then we introduce the ``extended'' indifferent drivers as follows. 
     \begin{definition}[{Extended indifferent driver/indifference curve}]
       Any point $(r, y)$ in $\mathbb{R}^2$ with $\Delta T^{i}(r, y) = 0$ is referred to as an extended indifferent driver.  The set of extended indifferent drivers is referred to as the extended indifference curve.
    \end{definition}
    
    With this definition,  an indifferent driver exists whenever the intersection of $R$ and the extended indifference curve is nonempty.  This yields a unified analysis framework for both  exogenous and endogenous waiting times; and both cases where indifferent drivers are present and absent.

    To solve the extended indifference curve, we let $\Delta T^{i}(r, y) = 0$, yielding
    \begin{align}\label{time difference}
        \psi (r, y) = w_{A}^{i} - w_{B}^{i} -\tau,
    \end{align}
    where we define
    $$\psi (r, y) = F(r-c+cy) - F(r-cy) - 2\tau y,$$ and $i = x, n$ represents the exogenous and endogenous models, respectively.
    
    Regarding the preceding equation, two questions arise:
    \begin{enumerate}[label=\textbf{Q\arabic*.}, ref=Q\arabic*]
        \item Does there exist a solution in $\mathbb{R}^2$? \label{q1}
        \item How to solve the equation? \label{q2}
    \end{enumerate}
    We observe that the LHS of equation \eqref{time difference} depends only on the charging characteristics and modeling parameters, while the RHS is dependent on    different waiting time models.

    For exogenous model, the answer to question \ref{q1} is not obvious, but for question \ref{q2} is straightforward. Equation \eqref{time difference} can be written as $\psi (r, y) = w_{A}^{x} - w_{B}^{x} -\tau$. As the RHS is a constant independent of other EV drivers' decisions, for each $r$, we can search for a solution of $y$. 

    For endogenous model, both question \ref{q1} and \ref{q2} are hard to answer. In this case, equation \eqref{time difference} is written as $\psi (r, y) = w_{A}^{x} - w_{B}^{x} -\tau - \epsilon + 2\epsilon \alpha$. As the RHS includes the congestion $\alpha$ that is dependent on all EVs' strategies, it is challenging to solve the equation. Specifically, we need to know the congestion $\alpha$ to solve the extended indifference curve. Meanwhile, it is necessary to know the indifference curve before calculating the congestion  $\alpha$. This dependency or coupling effect creates the challenge of solving equation \eqref{time difference} for endogenous model.
    
    Motivated by addressing this dependency and unifying two waiting time models, we propose an equivalent reformulation using ordinary differential equation (ODE). The idea is as follows. Since we assume all EV drivers see the same congestion $\alpha$, the RHS of equation \eqref{time difference} remains the same as $r$ varies. Then, by taking the derivatives of both sides of equation \eqref{time difference} with respect to $r$, the term involving $\alpha$ goes away and thus the aforementioned ``coupling effect''  disappears. Therefore, for an extended indifferent driver, we find $y$ as a function of $r$ by studying the ODE as follows,
    \begin{align}\label{ODE}
        \Big( cF'(r-c+cy) + cF'(r-cy) - 2\tau \Big)\frac{dy}{dr} \nonumber\\
        = - \Big( F'(r-c+cy) - F'(r-cy) \Big).
    \end{align}
    By construction, we set the derivative of $\psi$ to $0$ and see that the general solution of the preceding ODE has the form $\psi(r,y) = k$, where $k$ is an arbitrary constant. Combining with an appropriate initial condition, equation \eqref{time difference} can be solved. 

    {For the remainder of this section, we first present preliminary results. Then, Theorem \ref{exogenous optimal strategy}-\ref{endogenous optimal strategy} prove the existence of extended indifference curve and give the appropriate initial conditions to ODE in \eqref{ODE} for exogenous and endogenous models, respectively.}

\subsection{Preliminaries}
In this {subsection}, we provide some basic properties related to the previous ODE, which can be used in both exogenous and endogenous models.

   Lemma \ref{monotonicity} shows the monotonicity of function $\psi$. Lemma \ref{ODE general solution property} proves that the general solution of ODE in \eqref{ODE} is a well-defined  function. In Lemma \ref{IVP solution}, we formally introduce the initial value problem \eqref{IVP} and then characterize the solutions to such problem. {Lemma \ref{lemma-limit existence}-\ref{continuity of g(r,z)}} further describe some properties regarding the general solution of ODE. 
    
\begin{lemma}\label{monotonicity}
        Given any $r \in \mathbb{R}$, $\psi(r, \cdot)$ is strictly decreasing in $\mathbb{R}$.
    \end{lemma}

    \begin{proof}
        Since $F$ is differentiable, we observe that $\psi(r, y)=F(r-c+cy) - F(r-cy) - 2\tau y$ is also differentiable. Then, we take its partial derivative with respect to $y$ and obtain
        \begin{align*}
            \frac{\partial \psi(r, y)}{\partial y} = cF'(r-c+cy) + cF'(r-cy) - 2\tau.
        \end{align*}
        With $F$ being strictly decreasing, its first order derivative is negative and so on we have $\partial \psi(r, y) / \partial y < 0$ for any $y \in \mathbb{R}$. Hence, $\psi(r, \cdot)$ is strictly decreasing.
    \end{proof}

    Recall that the general solution to ODE in \eqref{ODE}, if exists, is given by $\psi(r,y) = k$, where $k$ is an arbitrary constant.  In the following lemma, we show the existence and uniqueness of such a solution.
    
    \begin{lemma}\label{ODE general solution property}
        For any constant $k \in \mathbb{R}$, there exists one and only one value $y \in \mathbb{R}$ such that $\psi(r,y) = k$ for any $r \in \mathbb{R}$.
    \end{lemma}

    \begin{proof}
        To show the existence,
        we observe that $F$ is continuous and strictly decreasing, we then have for all $r$, $\lim_{y \to \infty} \psi(r,y) = - \infty < 0$, $\lim_{y \to -\infty} \psi(r,y) = \infty > 0$. Therefore, by the intermediate value theorem, there exists a value $y \in \mathbb{R}$ such that  $\psi(r,y) = k$.

        The uniqueness of $y$ follows from the strict monotonicity of $\psi(r,\cdot)$ as established in  Lemma \ref{monotonicity}.
    \end{proof}

    
    From Lemma \ref{ODE general solution property}, we learn that the general solution of ODE in problem \eqref{IVP} is given by $\psi(r,y) = k$ with an arbitrary constant $k$. It implicitly defines a function mapping $g$ such that $y = g(r)$. The initial condition is naturally given by that constant $k \in \mathbb{R}$. However, this form of initial condition does not give any insight about the extended indifference curve, we here introduce an alternative initial condition $y(r_t) = z$ with $z \in \mathbb{R}$. For the rest of the paper, we focus on this formulation.

    \begin{lemma}\label{IVP solution}
        Consider the following initial value problem
        \begin{align*}\tag{IVP}\label{IVP}
            &\Big( cF'(r-c+cy) + cF'(r-cy) - 2\tau \Big)\frac{dy}{dr} \nonumber\\
            &= - \Big( F'(r-c+cy) - F'(r-cy) \Big),\\
            &y(r_t) = z.
        \end{align*}
        The initial conditions $y(r_t)=z$ with $z \in \mathbb{R}$ and  $\psi(r,y)=k$ with $k \in \mathbb{R}$ are equivalent.
        
    \end{lemma}

    \begin{proof}

        By Lemma \ref{monotonicity}, $\psi(r_t,z)$ is strictly decreasing in $z \in \mathbb{R}$. Hence, $\psi(r_t, \cdot)$ is a one-to-one mapping function. Since $\lim_{z \to \infty} \psi(r_t,z) = -\infty$ and $\lim_{z \to -\infty} \psi(r_t,z) = \infty$, function $\psi(r_t, \cdot)$ maps from $\mathbb{R}$ to $\mathbb{R}$. Therefore, every value of $z \in \mathbb{R}$ corresponds to a distinct value of $k \in \mathbb{R}$ with $k = \psi(r_t, z)$ and vice versa. Hence, we conclude that the two initial conditions $\psi(r, y) = k$ and $y(r_t) = z$ are equivalent. 
    \end{proof}

The preceding results suggest that there is a unique solution $y = g(r)$, where $g$ is a function well defined in $r \in \mathbb{R}$ and is determined implicitly by $\psi(r,y) = \psi(r_t,z)$. We note that we may use any vertical slice $y(\cdot)$ as the initial condition of problem \eqref{IVP}. We will discuss the reason of choosing $y(r_t)$ in Remark \ref{remark: selection of initial condition}.

 It is worth noting that the more general Picard's existence and uniqueness theorem for  first order initial value problems is not applicable, as it requires additional smoothness assumption on function $F$ \cite{Picard}.

    \begin{corollary}\label{geometric intuition}
        As the initial value $z$ varies, the curves of any two solutions of problem \eqref{IVP} cannot intersect each other.
    \end{corollary}

    \begin{proof}
         By Lemmas \ref{ODE general solution property}-\ref{IVP solution}, if any two solutions intersect, there will be two solutions that satisfy problem \eqref{IVP} with the same initial condition corresponding to the intersection point, which contradicts to the uniqueness.
    \end{proof}

    By regarding the initial value $z$ as an independent variable, the solution of problem \eqref{IVP} is also a function of $z$. Therefore, for the rest of the paper, we abuse the notation and write $y$ as a function of both $r$ and $z$, i.e., $y = g(r, z)$. {In Lemma \ref{lemma-limit existence}-\ref{continuity of g(r,z)}, we show that $g(r, z)$ is continuous in $z$.}

    \begin{lemma}\label{lemma-limit existence}
        With a fixed $r \in \mathbb{R}$,  
        {$g(r, \cdot)$ is non-decreasing} and $\lim_{z \to a} g(r, z)$ exists for any $a \in \mathbb{R}$.
    \end{lemma}

    \begin{proof}

        { 
        We first prove the monotonicity. Clearly $g(r, \cdot)$ is non-decreasing when $r = r_t$, as the initial value of problem \eqref{IVP} is $g(r_t, z) = z$. For the case of $r \neq r_t$, we assume the contrary that there exists $z_1, z_2 \in \mathbb{R}$ with $z_2 > z_1$ such that $g(r, z_1) > g(r, z_2)$. Denote $\Delta g(r) = g(r, z_1) - g(r, z_2)$. Since $g(r,z)$ is continuous with respect to $r$, $\Delta g(r)$ is continuous in $\mathbb{R}$. Hence we have $\Delta g(r) > 0$ and $\Delta g(r_t) < 0$. By intermediate value theorem, there exists a $r'$ between $r$ and $r_t$ such that $\Delta g(r') = g(r', z_1) - g(r', z_2) = 0$. However, this contradicts Corollary \ref{geometric intuition}, which shows that any two solutions of problem \eqref{IVP} cannot intersect each other. Therefore, $g(r, \cdot)$ is non-decreasing.
        
        Next, we prove the existence of the limit. The monotonicity implies that, with any fixed $r \in \mathbb{R}$, the one-sided limits $\lim_{z \to a-} g(r, z)$ and $\lim_{z \to a+} g(r, z)$ exist for any $a \in \mathbb{R}$ \cite{rudin}. Then, $\lim_{z \to a} g(r, z)$ exists if and only if the two limits are equal.

        }

        We assume that $\lim_{z \to a-} g(r, z) \neq \lim_{z \to a+} g(r, z)$ for some $r$ and $a$. According to the strict monotonicity of $\psi(r, \cdot)$ shown in Lemma \ref{monotonicity}, there holds
        $$k_1=\psi(r, \lim_{z \to a-} g(r, z)) \neq \psi(r, \lim_{z \to a+} g(r, z)) = k_2,$$ for some scalars $k_1, k_2$.
        Since $y = g(r, z)$ is determined implicitly by $\psi(r, y) = \psi(r_t, z)$, we have $z = g(r_t, z)$ and thus the following inequality,
        $$k_1=\psi(r_t, \lim_{z \to a-} z) \neq \psi(r_t, \lim_{z \to a+} z)=k_2.$$
        Again, by Lemma \ref{monotonicity}, the preceding inequality implies that $\lim_{z \to a-} z \neq \lim_{z \to a+} z$, which is not true. Therefore, $\lim_{z \to a} g(r, z)$ exists for any $r, a \in \mathbb{R}$.
    \end{proof}

    \begin{lemma}\label{lemma-limit transitivity}
        For any $r, a \in \mathbb{R}$, the following equation holds,
        $$\lim_{z \to a} \psi(r,g(r,z)) = \psi(r,\lim_{z \to a} g(r,z)).$$
    \end{lemma}

    \begin{proof}
        {By continuity of function $\psi$ and the existence of $\lim_{z \to a} g(r,z)$, the statement can be established.}    
    \end{proof}

    \begin{lemma}\label{continuity of g(r,z)}
        For any $r \in \mathbb{R}$, function $y = g(r,z)$ is continuous with respect to $z$.
    \end{lemma}

    \begin{proof}
        For any initial value $z \in \mathbb{R}$, we have $\psi(r,g(r,z)) = \psi(r_t,z)$. Suppose we add a small perturbation $\Delta$ around $z$, then, the following equation holds for any $r \in \mathbb{R}$,
        \begin{align*}
            &\psi(r,g(r,z+\Delta)) = \psi(r_t,z+\Delta)\\
            &= F(r_t-c+cz+c\Delta) - F(r_t-cz-c\Delta) - 2\tau(z+\Delta).
        \end{align*}
        By the continuity of $F$, we have
        \begin{align*}
            &\lim_{\Delta \to 0} F(r_t-c+cz+c\Delta) = F(r_t-c+cz),\\
            &\lim_{\Delta \to 0} F(r_t-cz-c\Delta) = F(r_t-cz),
        \end{align*}
        which imply
        \begin{align*}
            \lim_{\Delta \to 0} \psi(r_t,z+\Delta) &= F(r_t-c+cz) - F(r_t-cz) - 2\tau z\\
            &= \psi(r_t,z).
        \end{align*}
        Therefore,  
        \begin{align*}
            \lim_{\Delta \to 0} \psi(r,g(r,z+\Delta)) = \psi(r_t,z) = \psi(r,g(r,z)).
        \end{align*}
        From Lemma \ref{lemma-limit transitivity}, the preceding equation can be rewritten as
        $$\psi(r,\lim_{\Delta \to 0} g(r,z+\Delta)) = \psi(r,g(r,z)).$$
        According to Lemma \ref{monotonicity}, $\psi(r,\cdot)$ is strictly decreasing in $\mathbb{R}$. It follows that
        \begin{align*}
            \lim_{\Delta \to 0} g(r,z+\Delta) = g(r,z),
        \end{align*}
        thus $g(r,z)$ is continuous in $z \in \mathbb{R}$.
    \end{proof}

    {To summarize we show that the solution of problem \eqref{IVP} depends continuously on the initial value $z$. \footnote{Note that though there is existing ODE theorem  \cite{arnold} implying a similar continuous dependence of an initial value problem, it requires additional smoothness conditions on $F$, and hence is not applicable here.}}

    \subsection{Exogenous Waiting Time Model}
     In this subsection, we present the following theorem for exogenous model. It shows the existence of extended indifference curve. Combine with Lemma \ref{IVP solution}, the question \ref{q1} and \ref{q2} are answered. Furthermore, the optimal strategies for all EV drivers in region $R$ are derived.

\begin{theorem}\label{exogenous optimal strategy}
        In the exogenous waiting time model, there exists an initial value $z$ satisfying
        \begin{align}\label{exogenous initial condition}
            \psi(r_t,z) =  w_A^x - w_B^x - \tau.
        \end{align}
        Denote the corresponding solution of problem \eqref{IVP} by $g(r,z)$. For any $r$ in the region $R$, the optimal strategy of EV associated with the pair $(r,y)$ is as follows:
        \begin{itemize}
            \item Choose Station A if $y < g(r,z)$;
            \item Choose Station B if $y > g(r,z)$;
            \item Be indifferent if $y = g(r,z)$.
        \end{itemize}
    \end{theorem}

    \begin{proof}
        First, as the right hand side of equation \eqref{exogenous initial condition} is a constant, such initial value $z$ satisfying equation \eqref{exogenous initial condition}  exists and is unique.
        
        Using the definition of total time cost in Section \ref{model}, we then obtain $\Delta T^x(r,y)  = \psi(r,y) + \tau + w_B^x - w_A^x$ for any $r$.  With $z$ satisfying equation \eqref{exogenous initial condition}, we have $\Delta T^x(r,g(r,z))= 0$. Lemma \ref{monotonicity} implies that $\psi(r, y)$ and  thus $\Delta T^x(r,y)$ are strictly decreasing in $y \in \mathbb{R}$. If $y < g(r,z)$, then $\Delta T^x(r,y)  > 0$, implying that the optimal strategy is to choose Station A.
          On the other hand, if $y > g(r,z)$, then $\Delta T^x(r,y)< 0$, so the optimal choice is to opt for Station B. The case of $y = g(r,z)$ implies that $\Delta T^x(r,y) = 0$, meaning the EV driver is indifferent between two charging stations.
    \end{proof}

    The function $g(r,z)$ stated in Theorem \ref{exogenous optimal strategy} is the extended indifference curve. If the intersection of the region $R$ and $g(r,z)$ is empty, all EV drivers in $R$ strictly prefer the same charging station. Otherwise, there exists at least one driver in $R$ who is indifferent between the two charging stations and people on either side of the indifference curve choose different charging stations. In Section \ref{numerical exogenous model} we illustrate some examples of indifference curves by numerical simulations.

    \subsection{Endogenous Waiting Time Model}

        In the aforementioned exogenous model, EV drivers make decisions without taking into account the strategies of other drivers. However, a more rational decision maker will also respond to the strategies of other drivers. Therefore, we further study the endogenous model, where EV drivers consider not only the exogenous waiting times, but also the congestion at two charging stations. We assume that the parameter pairs $(r, y)$ of EV drivers are uniformly distributed in the region $R$.

  Before delving into the main theorem for endogenous model, we need to establish another preliminary result related to the monotonicity and continuity of an integration. Specifically, we define 
    \begin{align*}
        A(z)  = \frac{1}{r_t-c} \int_{c}^{r_t} h(r,z) dr,
    \end{align*}
    where $h(r,z) = \min\{ 1, \max\{0, g(r,z)\} \}$.

  The integration $A(z)$ can be interpreted as the integral of $g(r, z)$ intersecting region $R$.  In the following lemma, we show the properties of this integration. 

    \begin{lemma}\label{Area}
        $A(z)$ is a bounded, non-decreasing and continuous function.
    \end{lemma}

    \begin{proof}
        According to the definition of $h(r, z)$, $A(z)$ is naturally lower bounded by 0 and upper bounded by 1.
    
        The monotonicity follows from Lemma \ref{lemma-limit existence}.

        To prove the continuity, we use the property that $g(r, z)$ is continuous in $z$ by Lemma \ref{continuity of g(r,z)}. That says, for any  $r \in \mathbb{R}$ and  $z \in \mathbb{R}$ and every $\varepsilon > 0$, there exists a $\delta > 0$ such that
        \begin{align*}
            |h(r,z) - h(r,z')| < \varepsilon,
        \end{align*}
        for all points $z' \in \mathbb{R}$ satisfying $|z - z'| < \delta$.

        Then we have
        \begin{align*}
            |A(z) - A(z')| &= \frac{1}{r_t-c} \Big| \int_{c}^{r_t} h(r,z) dr - \int_{c}^{r_t} h(r,z') dr \Big|\\
            &\le \frac{1}{r_t-c} \int_{c}^{r_t} \Big| h(r,z) - h(r,z') \Big| dr\\
            &< \frac{1}{r_t-c}\varepsilon (r_t-c)\\
            &= \varepsilon.
        \end{align*}
        Hence, $A(z)$ is continuous in $z$.
    \end{proof}

With the proper initial value $z$,  $g(r, z)$ is the extended indifference curve, then the integration $A(z)$ represents the congestion $\alpha$. Formally, we first give the optimal strategies of EV drivers in the next theorem, assuming the extended indifference curve exists. Then, we find the appropriate initial value $z$ such that this existence is guaranteed. In this way, the two questions \ref{q1} and \ref{q2} are naturally answered.

    \begin{theorem}\label{thm: endogenous strategy}
        If $g(r, z)$ is the extended indifference curve, the optimal strategy of EV associated with the pair $(r,y)$ is as follows:
        \begin{itemize}
            \item Choose Station A if $y < g(r,z)$;
            \item Choose Station B if $y > g(r,z)$;
            \item Be indifferent if $y = g(r,z)$.
        \end{itemize}
    \end{theorem}

    \begin{proof}
        We first write the total time cost defined in Section \ref{model},
        \begin{align*}
            \Delta T^n(r,y) = \psi(r,y) + \tau + w_B^x - w_A^x + \epsilon - 2\epsilon\alpha,
        \end{align*}
        where $\alpha \in [0, 1]$ represents the proportion of EVs that choose Station A as their optimal strategy. If $g(r,z)$ is indeed the extended indifference curve, it satisfies $\Delta T^n(r,g(r,z)) = 0$. Since we assume all EV drivers see the same congestion when making decisions, Lemma \ref{monotonicity} implies that $\Delta T^n(r,y)$ is strictly decreasing in $y$. Therefore, we have the following observations. If $y < g(r,z)$, then $\Delta T^n(r,y)  > 0$, implying that the optimal strategy is to choose Station A. On the other hand, if $y > g(r,z)$, then $\Delta T^n(r,y)< 0$, so the optimal strategy is to choose Station B. The case of $y = g(r,z)$ implies that $\Delta T^n(r,y) = 0$, meaning the EV driver is indifferent between two charging stations.
    \end{proof}

    \begin{theorem}\label{endogenous optimal strategy}
        In the endogenous waiting time model, there exists an initial value $z$ satisfying
        \begin{align}\label{endogenous initial condition}
            \psi(r_t,z) = w_A^x - w_B^x - \tau - \epsilon + {2\epsilon}A(z),
        \end{align}
        such that the corresponding solution of problem \eqref{IVP}, denoted by $g(r,z)$, is the extended indifference curve.
    \end{theorem}

    \begin{proof}

        First, we show that there exists one and only one value of $z \in \mathbb{R}$ such that equation \eqref{endogenous initial condition} is satisfied. Denote
        \begin{align*}
            \rho(z) = \psi(r_t,z) + \tau + w_B^x - w_A^x + \epsilon - 2\epsilon A(z).
        \end{align*}
        By Lemma \ref{monotonicity} and Lemma \ref{Area}, we observe that $\rho(z)$ is continuous and strictly decreasing. As $A(z)$ is bounded, we have $\lim_{z \to -\infty} \rho(z) = \infty$ and $\lim_{z \to \infty} \rho(z) = -\infty$. Therefore, by the intermediate value theorem, there exists an initial value $z \in \mathbb{R}$ such that $\rho(z) = 0$ and so on equation \eqref{endogenous initial condition} is satisfied. The uniqueness follows from the strict monotonicity of $\rho(z)$.

        By letting $\rho(z) = 0$, we observe that $g(r,z)$ is the extended indifference curve satisfying $\Delta T^n(r,g(r,z)) = 0$, while the integral $A(z)$ measures the proportion of drivers who choose Station A, i.e., the congestion $\alpha$.
    \end{proof}

    {Finally, we conclude this section by comparing with our preliminary work \cite{allerton}. In this paper, we first extend region $R$ and introduce the concept of extended indifferent driver. Then we propose a unified ODE-based formulation for both exogenous and endogenous models, prove the existence of extended indifference curve and give the appropriate initial conditions to the ODE. In comparison, our preliminary work focuses on the exogenous model. It first derives the existence conditions of indifferent drivers in region $R$, and then analyzes the properties of indifference curve  under specific assumptions. However, that method cannot be directly extended to endogenous model due to the dependency of congestion and strategies. The advantage of ODE approach is that it provides an efficient way to solve the endogenous model and gives a unified framework for both waiting time models. In fact, the exogenous model can be regarded as a special case of endogenous model by letting $\epsilon=0$.}

    {Since we extend the region for analysis convenience, the resulting extended indifference curve may lie outside the physically feasible region $R$. In the next section, we  characterize the conditions for which it is in $R$.}



\section{Characteristics of the Extended Indifference Curve}\label{sec: decreasing charging rates}

    {Recall that we assume the charging rate function $P$ to be continuous and the charging time function $F$ is defined as $F(r) = \int_{r}^{r_t} E/P(\gamma) d\gamma$.} In Section \ref{Optimal strategies}, we consider the charging rate function $P$ in its general form, with the only assumption being its continuity. In practice, as shown in Fig. \ref{field data charging rates}, EVs' charging rate is typically non-increasing as the state-of-charge increases. Therefore, we formally make the following assumption in this section.

    \begin{assumption}\label{assumption: P non-increasing}
        We consider charging rate function $P$ to be non-increasing in $[0, 1]$.
    \end{assumption}

    \begin{lemma}
        The charging time function $F$ is concave.
    \end{lemma}
    \begin{proof}
        Since $F$ is differentiable, we take the first order derivative as $F'(r) = -\frac{E}{P(r)}$. By definition and Assumption \ref{assumption: P non-increasing}, $P(r)$ is strictly positive and non-increasing. Hence $F'(r)$ is (weakly) decreasing and so on $F$ is concave.
    \end{proof}

    We then identify the weak monotonicity of the extended indifference curve under Assumption \ref{assumption: P non-increasing}.

    \begin{theorem}\label{concave property}
        For both exogenous and endogenous models, the solution $y = g(r,z)$ of problem \eqref{IVP} has the following properties for all $r \in \mathbb{R}$:
        \begin{itemize}
            \item If $z = 1/2$, $g(r,z) = 1/2$;
            \item If $z < 1/2$, $g(r,z) < 1/2$ and is non-decreasing in $r$;
            \item If $z > 1/2$, $g(r,z) > 1/2$ and is non-increasing in $r$.
        \end{itemize}
    \end{theorem}

    \begin{proof}
        When $z = 1/2$, we verify that $y = g(r,z) = 1/2$ is a solution of problem \eqref{IVP} by observing that $r-c+cy = r-cy=r-\frac{c}{2}$ and the ODE reduces to $\frac{dy}{dr}=0$. According to Lemma \ref{ODE general solution property} and \ref{IVP solution}, the solution is unique. Thus, we have $g(r,z) = 1/2$ for all $r \in \mathbb{R}$.

        If $z < 1/2$, by Corollary \ref{geometric intuition}, $g(r,z) < 1/2$ for all $r \in \mathbb{R}$. Otherwise, it will intersect the solution with the initial value being $z = 1/2$. To demonstrate the monotonicity of $y = g(r,z)$ with respect to $r$, it is equivalent to showing $dy/dr \ge 0$ for any $r \in \mathbb{R}$. We assume the contrary that there exists one $r$ such that $dy/dr < 0$. {As $y = g(r, z) < 1/2$, we have $r-cy > r-c+cy$}. Then the LHS of the ODE in \eqref{ODE} is positive since $F$ is strictly decreasing, which contradicts the RHS being non-positive due to the concavity of $F$. Therefore, $g(r,z)$ is non-decreasing in $r$ when $z < 1/2$.

        The case of $z > 1/2$ can be established similarly.
    \end{proof}

    {Intuitively, as Station A is at $y=0$ and Station B is at $y=1$, the two charging stations are symmetric around $y=1/2$. If $z=1/2$, there is an indifferent driver at $(r_t, 1/2)$. In this case, even if we exchange the positions of Station A and B, the extended indifference curve should remain the same. Thus, we shall have $g(r,z) = 1/2$ for all $r \in \mathbb{R}$ as stated in Theorem \ref{concave property}. On the other hand, if $z < 1/2$, $g(r,z)$ should also be symmetric with respect to the line $y=1/2$, compared to the case when $z > 1/2$. Furthermore, we see that the extended indifference curve is weakly monotonic by assuming the monotonicity of charging rates. Then, we show the necessary and sufficient condition under which an indifferent driver exists in region $R$.}

    \begin{corollary}
        In the exogenous model, an indifferent driver exists in region $R$ if and only if the following condition holds
        \begin{align}
            F(r_{t}-c) - F(r_{t}) \ge - \tau + |w_{A}^x - w_{B}^x|. \label{exogenous - indifferent driver existence condition}
        \end{align}
    \end{corollary}

    \begin{proof}
        Based on Theorem \ref{exogenous optimal strategy} and \ref{concave property}, we observe that an indifferent driver exists in $R$ if and only if the initial value $z$ satisfies both equation \eqref{exogenous initial condition} and $0 \le z \le 1$. Using Lemma \ref{monotonicity}, we find that $\psi(r_t, z)$ is strictly decreasing in $z$. Therefore, an equivalent characterization is       \begin{align*}
            \psi(r_t, 1) \le - \tau - w_B^x + w_A^x \le \psi(r_t, 0),
        \end{align*}
        which yields inequality \eqref{exogenous - indifferent driver existence condition}.
    \end{proof}

    With charging rate non-increasing as the state-of-charge increases, we can infer that when inequality \eqref{exogenous - indifferent driver existence condition} is violated for exogenous model, all EVs will choose the same charging station as their optimal strategies. This scenario is more likely to occur when one of the following conditions is met:
    \begin{enumerate}
        \item The waiting time difference, $|w_{A}^x - w_{B}^x|$, is large;
        \item The unit distance traveling time, $\tau$, is small;
        \item The power consumption per unit distance, $c$, is small;
        \item The target battery capacity level, $r_t$, is small.
    \end{enumerate}
    
    From an EV's perspective, two charging stations appear identical except for the waiting time difference. Hence, all EV drivers are inclined to select the station with the lower waiting time, especially when the difference is large. Furthermore, if the unit distance traveling time and the power consumption over unit distance approach zero, that is, $\tau, c \to 0$, i.e., it is cheap to travel,  then even a small waiting time difference may cause all EV drivers to choose the same charging station and leave the other station completely idle. The last condition is due to the concavity of $F$ as verified by real-world data that charging time from 20\% to 40\%  is much shorter than 80\% to 100\% (full battery). In this case, the overall cost of charging is dominated by the waiting/travel time if only charging to a smaller $r_t$.   Consequently, in the exogenous model, it is likely that all drivers head to one charging station with a shorter waiting time, whereas the rational  EV drivers should anticipate that the once-idle charging station would soon become congested. This phenomenon is alleviated in the endogenous model.

        \begin{corollary}\label{endogenous indifferent driver}
        In the endogenous model, an indifferent driver exists in region $R$ if and only if the following condition holds
        \begin{align}
            F(r_{t}-c) - F(r_{t}) \ge - \tau - \epsilon + |w_{A}^x - w_{B}^x|. \label{endogenous - indifferent driver existence condition}
        \end{align}
    \end{corollary}

    \begin{proof}
        The proof follows from Theorem \ref{endogenous optimal strategy} and \ref{concave property}, which indicates that an indifferent driver exists if and only if the initial value $z$ satisfies the conditions $0 \le z \le 1$ and equation \eqref{endogenous initial condition}. Denote
        \begin{align*}
            \rho(z) = \psi(r_t,z) + \tau + w_B^x - w_A^x + \epsilon - 2\epsilon A(z).
        \end{align*}
        Lemma \ref{monotonicity} and Lemma \ref{Area} ensure that $\rho(z)$ is continuous and strictly decreasing in $z$. {Therefore, an equivalent characterization is }
        \begin{align*}
            \rho(0) \ge 0 \ge \rho(1),
        \end{align*}
        where $A(0) = 0$ and $A(1) = 1$. Simplifying the preceding inequalities leads to inequality \eqref{endogenous - indifferent driver existence condition}.
    \end{proof}

    \begin{remark}\label{remark: selection of initial condition}
        In problem \eqref{IVP}, we could have chosen any other vertical slice $y(\cdot)$ as the initial condition rather than $y(r_t)$. However, $0 \le y(r_t) \le 1$ is the only ``tight'' condition for the existence of indifferent driver in $R$, due to the monotonicity of $g(r,z)$ shown in Theorem \ref{concave property}.
    \end{remark}
    
    
    Compared with inequality \eqref{exogenous - indifferent driver existence condition}, we observe that whenever there is an indifferent driver for exogenous model, there must be one for endogenous model as well. The reverse is not true. {The inequality \eqref{endogenous - indifferent driver existence condition} is more likely to be violated under the same four conditions discussed above for exogenous model. Furthermore, as $\epsilon$ decreases, the impact of congestion is reduced. Thus, the difference in exogenous waiting times, i.e., $|w_{A}^x - w_{B}^x|$, becomes more substantial and EV drivers are more likely to choose the same charging station with the lower exogenous waiting time. In Section \ref{numerical exogenous model} and \ref{numerical endogenous model}, we will numerically show the impact of parameters and compare the results of two waiting time models.}

\section{{Numerical Case Study for Exogenous Model}}\label{numerical exogenous model}

    {In this section, we use numerical case studies to showcase the optimal strategies that EVs with different locations and remaining battery levels adopt under the exogenous waiting time model}. Specifically, we examine how {charging rate functions}, {exogenous waiting times}, unit distance traveling time and power consumption over the unit distance impact the strategies employed by EVs. To illustrate the charging decisions, we depict the region $R$ as a rectangle in Cartesian coordinate system where $r$ is the horizontal axis and $y$ is the vertical axis. We plot the indifference curve that shows all the indifferent EV drivers. {According to Theorem \ref{exogenous optimal strategy}, the lower region below the curve contains all EVs that select Station A as their optimal strategies, while the upper region includes all EVs that choose Station B as optimal strategies.}

    \subsection{{The impact of charging rate characteristics}}
    We first examine the impact of charging rate characteristics as shown in Fig. \ref{different charging time functions}, where we set $c = 0.2$, $\tau = 1$, $r_{t} = 1$, $w_{A}^{x} = 0.5$ and $w_{B}^{x} = 0$. Fig. \ref{different charging time functions - charging rates} shows the (weakly) decreasing charging rate curves for each charging time function, while Fig. \ref{different charging time functions - indifference curve} shows the corresponding indifference curves. With an analytical expression of the charging time function, the indifference curve $g$ can be explicitly solved. {When charging rate is a constant, i.e., $F(r) = 1 - r$,} the indifference curve is a line parallel to $r$-axis. It divides the region $R$ into two parts around $y=0.29$. Here, we observe that fewer EVs use Station A than B. This is potentially due to the longer waiting time at Station A. Thus, some EVs close to Station A are willing to spend slightly increased traveling and charging times for an exchange of a much lower waiting time at Station B. 

    \begin{figure}[ht]
        \subfigure[Charging rate curves]{
            \begin{minipage}[t]{0.5\linewidth}
			\centering
			\includegraphics[width=\linewidth]{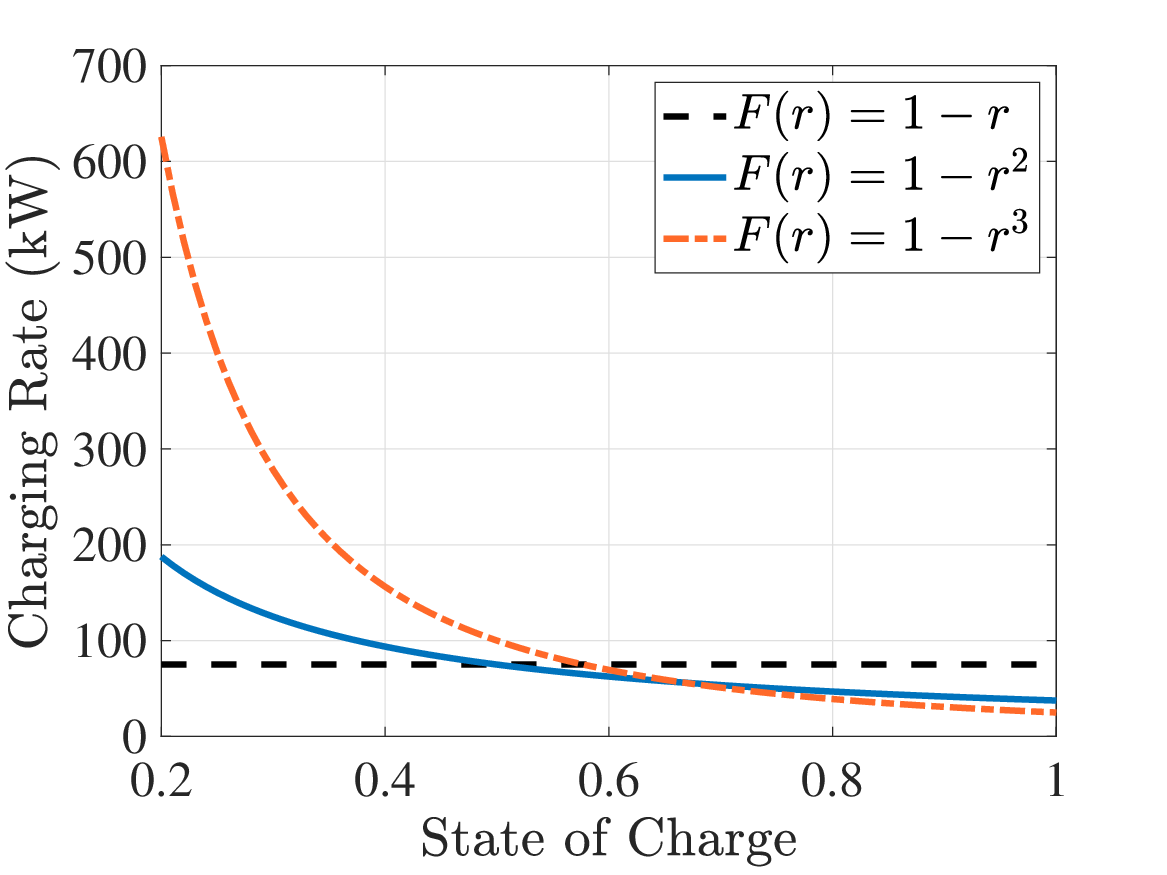}
			\label{different charging time functions - charging rates}
		\end{minipage}
        }%
        \subfigure[Indifference curves]{
            \begin{minipage}[t]{0.5\linewidth}
			\centering
			\includegraphics[width=\linewidth]{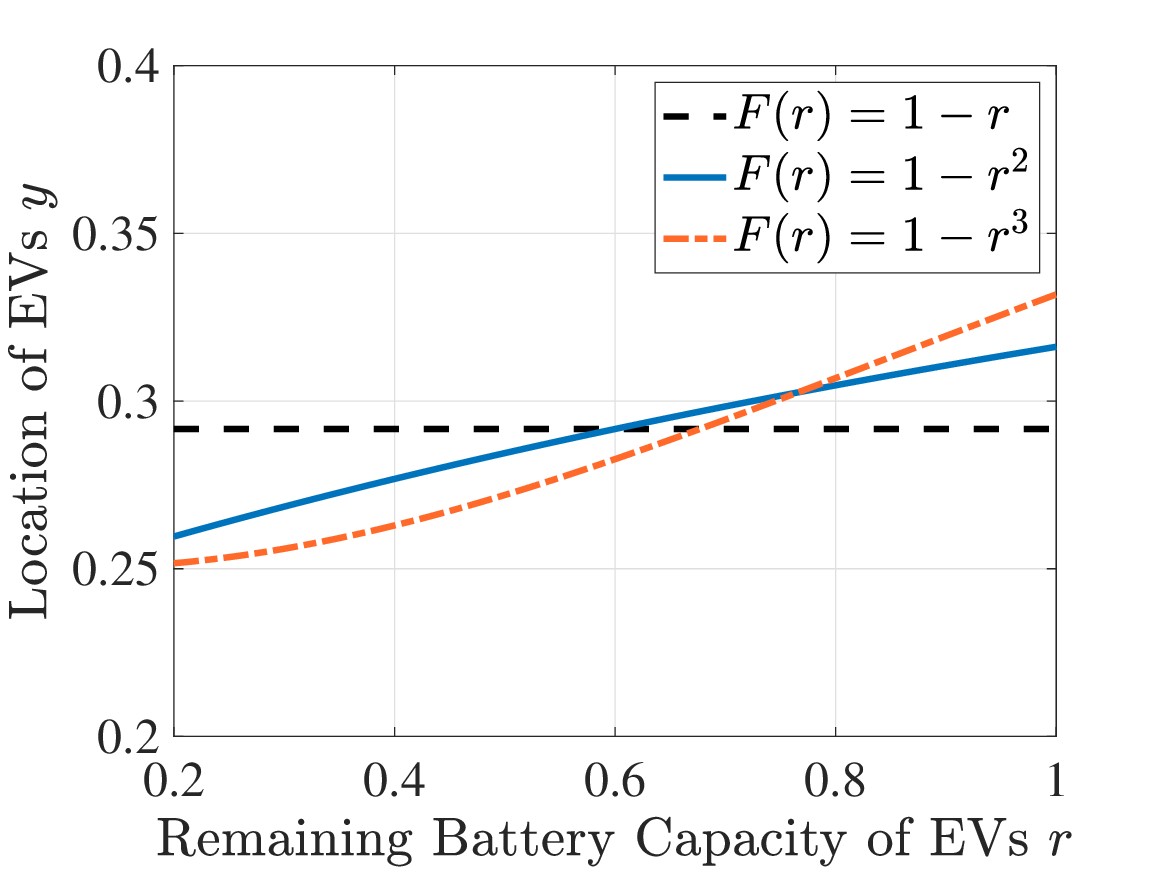}
			\label{different charging time functions - indifference curve}
		\end{minipage}
        }
        \caption{Comparison of charging rate curves and indifference curves {with different charging rate characteristics but the same average charging speed from $r=0$ to $r=1$.}}
        \label{different charging time functions}
    \end{figure} 
	
    However, {when charging rate is strictly decreasing, i.e., $F(r) = 1 - r^{2}$,} and other parameters remain the same, the indifference curve becomes an increasing function. For EVs located near the previous indifference curve with $y$ values close to $0.29$, since the charging speed is faster with low remaining battery levels, more EVs with small values of $r$ would be willing to travel a bit further and take Station B as their optimal strategies instead of Station A comparing to the linear charging time function case. What is happening here is that their charging time difference at two stations is much reduced due to the fast charging speed at low battery capacities. Meanwhile, the traveling times and waiting times for both stations remain the same as in the linear case. Therefore, such EVs tend to choose Station B instead. On the other hand, more EVs with large values of $r$ choose Station A as their optimal strategy rather than Station B. 
	This is due to the slower charging rate when the battery is nearly full. These vehicles opt for a closer charging station to balance the tradeoff between traveling time, charging time and waiting time. 
	
    {When charging rates change more dramatically with $r$, i.e., $F(r) = 1 - r^{3}$,} the aforementioned phenomena become more evident. Specifically, more EVs with small remaining battery levels opt for Station B, while more EVs with large remaining battery levels choose Station A, in comparison to both previously mentioned cases.

    It is worth noting that all charging time functions in this subsection take the same value of $F(0)$, meaning that the same amount of time is required to charge an EV from empty to full, regardless of the variation in instantaneous charging rates. As a result, the indifference curves of three cases are close to each other.


    \subsection{{The impact of average charging speed}}
    {Next, we further consider the impact of average charging speed that EV has to charge from empty to full, which can be reflected by $F(0)$.} The parameters are set as $c = 0.2$, $\tau = 1$, $r_{t} = 1$, $w_{A}^{x} = 0.5$ and $w_{B}^{x} = 0$. Fig. \ref{different charging times - charging rates} illustrates the charging rate curves and Fig. \ref{different charging times - indifference curve} shows the indifference curves. Specifically, the charging time functions $F(r) = 8(1-r^{2})$ and $F(r) = 4(1-r^{2})$ model the real-world scenario of using a level 2 charging, while the remaining cases correspond to a DC fast charging or supercharging \cite{charging rate at 90 percent}. As the overall charging process speeds up, the indifference curve shifts downwards and more EV drivers tend to choose Station B as their optimal strategy, regardless of the remaining battery levels. This is because the charging time difference between two stations decreases as the charging time is generally reduced. It results in more EVs opting to choose a station far away but with a significantly lower waiting time.

    \begin{figure}[ht]
        \subfigure[Charging rate curves]{
            \begin{minipage}[t]{0.5\linewidth}
			\centering
			\includegraphics[width=\linewidth]{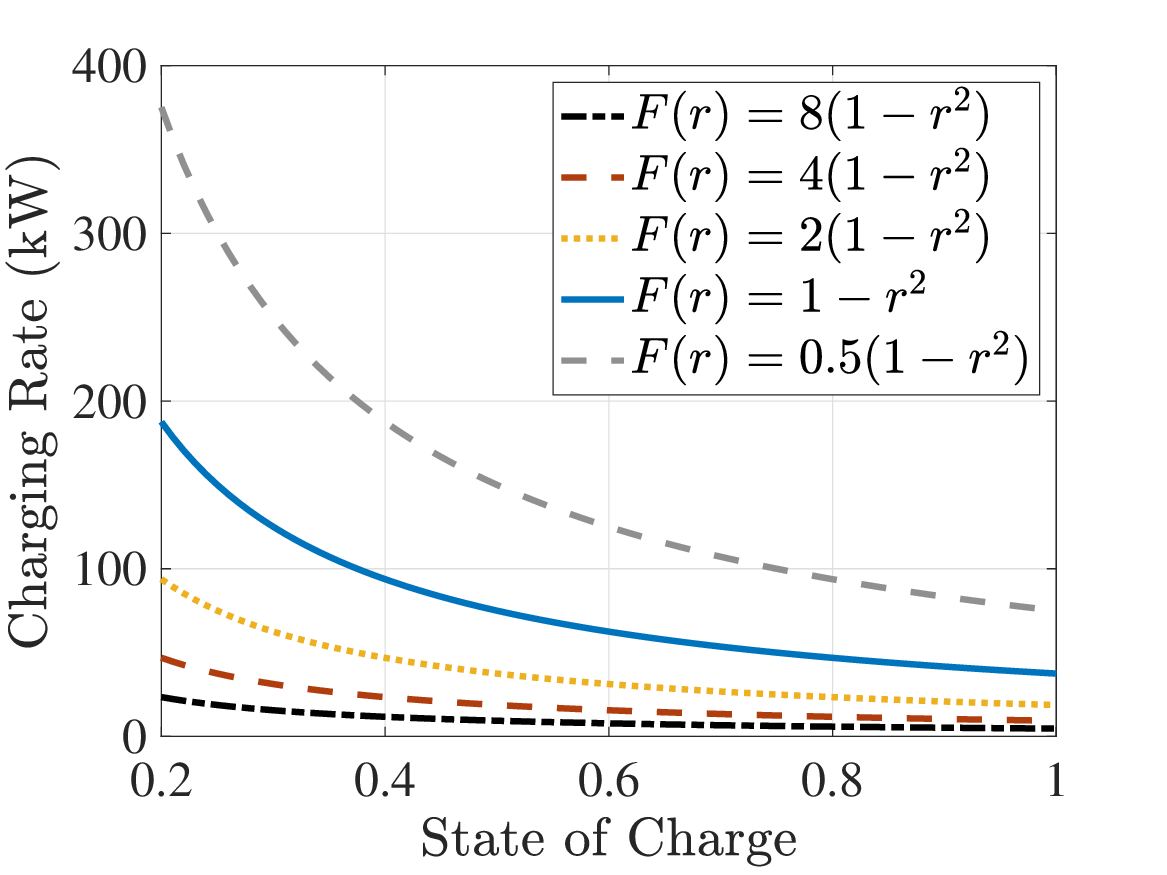}
			\label{different charging times - charging rates}
		\end{minipage}
        }%
        \subfigure[Indifference curves]{
            \begin{minipage}[t]{0.5\linewidth}
			\centering
			\includegraphics[width=\linewidth]{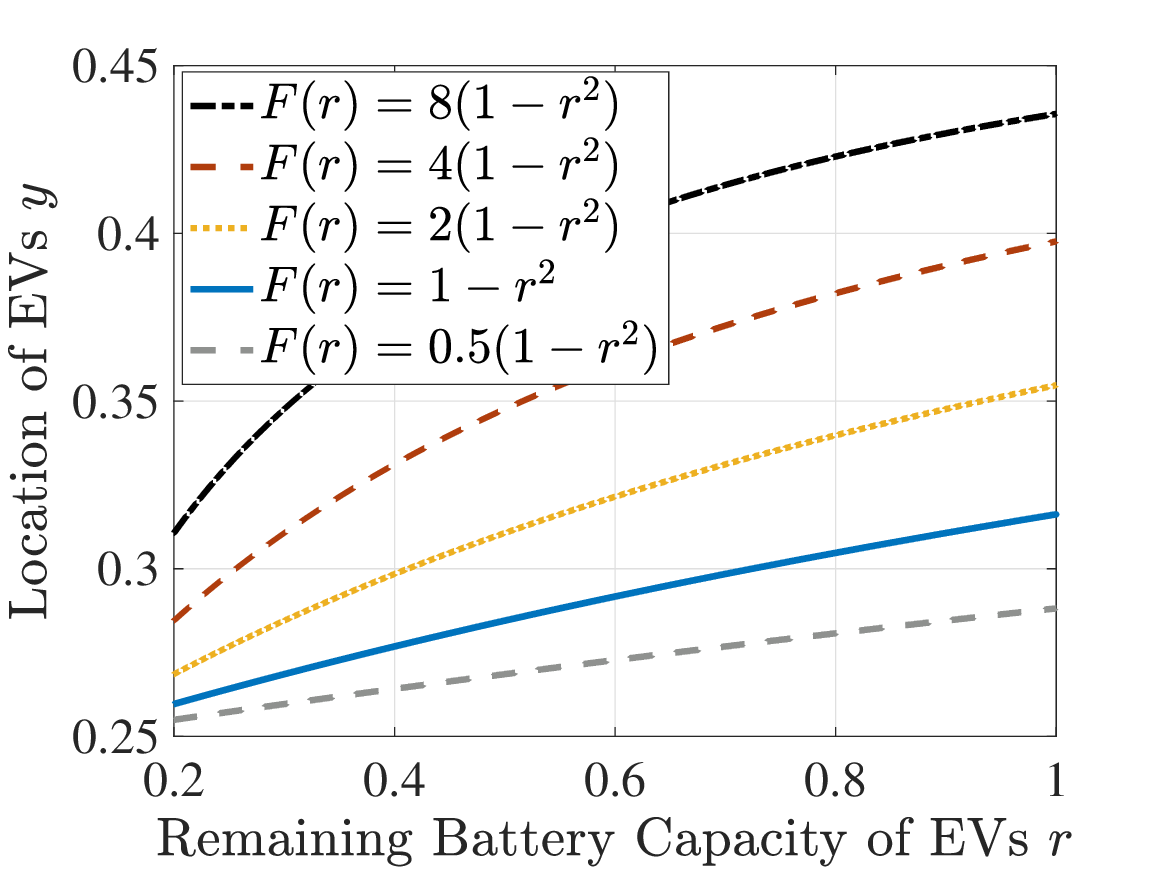}
			\label{different charging times - indifference curve}
		\end{minipage}
        }
        \caption{Comparison of charging rate curves and indifference curves {with different charging rate characteristics and different average charging speeds from $r=0$ to $r=1$.}}
        \label{different charging times}
    \end{figure}

    
    

    \subsection{{The impact of modeling parameters}}
    In this subsection, we investigate the impact of modeling parameters on the optimal strategies of EVs, including exogenous waiting times at charging stations, unit distance traveling time and power consumption over unit distance. The corresponding indifference curves are depicted in Fig. \ref{different parameters}. Before delving into the details, we set $r_{t} = 1$, $F(r) = 1 - r^{2}$ and $w_{B}^{x} = 0$ as unchanged parameters. In the first case, we assume $c = 0.2$, $\tau = 1$ and $w_{A}^{x} = 0.5$.

    \begin{figure}[ht]
    \begin{minipage}[b]{0.48\linewidth}
    \centering
    \includegraphics[width=\textwidth]{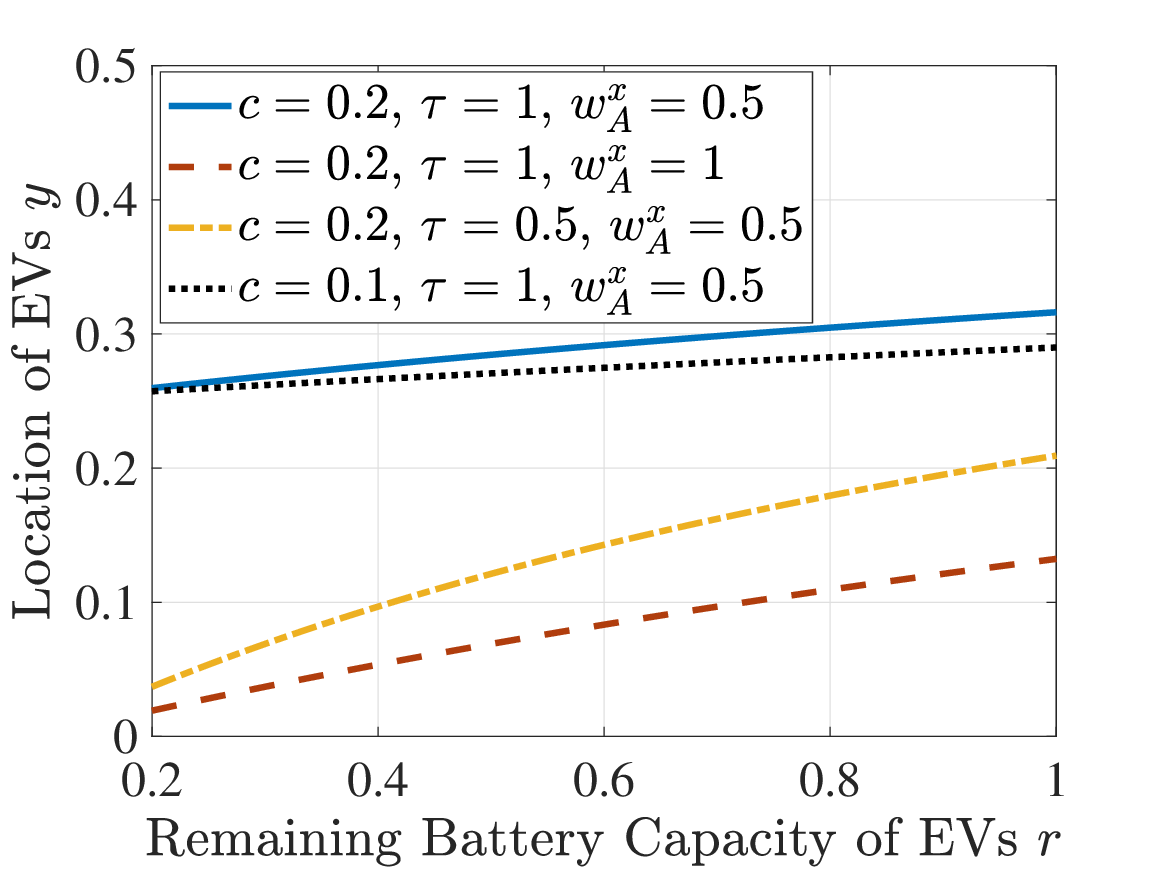}
    \caption{Indifference curves with different modeling parameters.}
    \label{different parameters}
    \end{minipage}
    \hspace{0.1cm}
    \begin{minipage}[b]{0.48\linewidth}
    \centering
    \includegraphics[width=\textwidth]{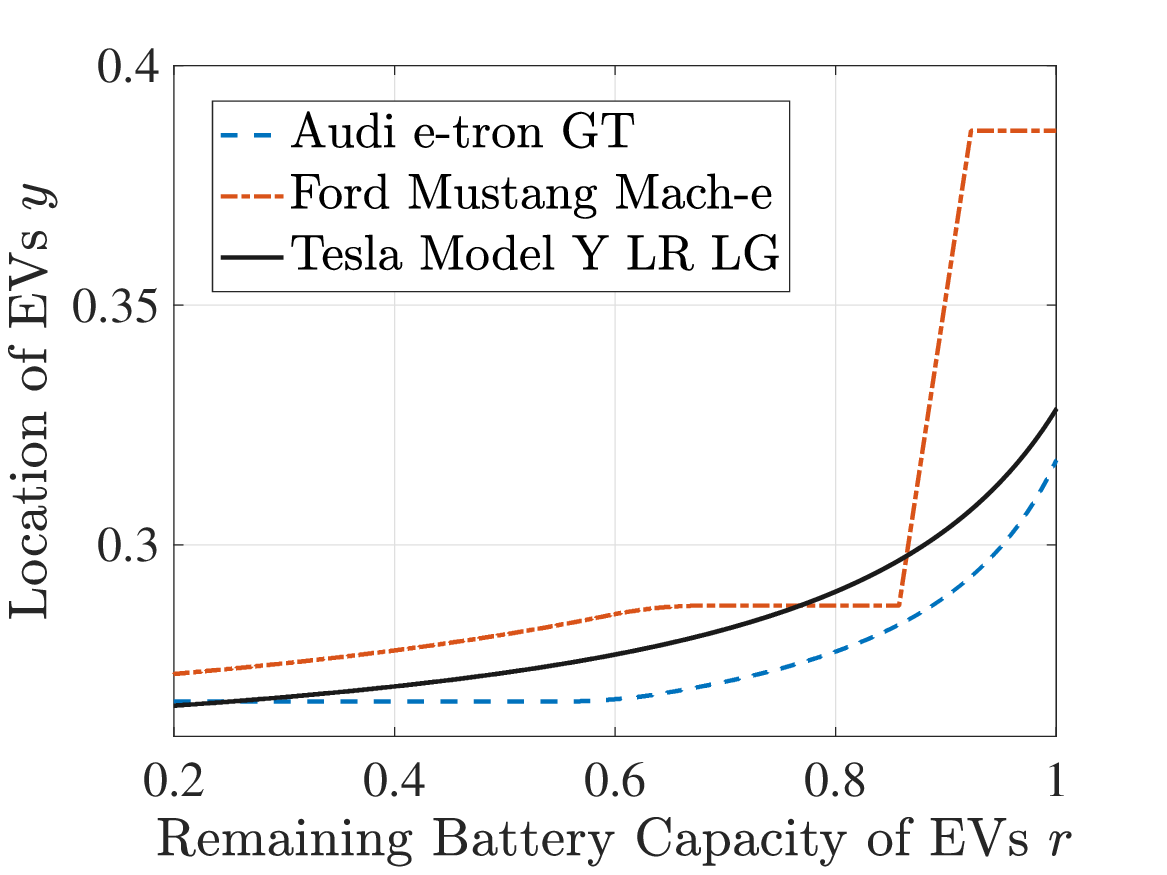}
    \caption{Indifference curves using field data of EV's charging rates.}
    \label{field data indifference curves}
    \end{minipage}
    \end{figure}

	
	Subsequently, we examine the scenario where the waiting time at Station A is longer, i.e., $w_{A}^{x} = 1$, while the other parameters remain the same. In this case, the indifference curve shifts downwards and gets closer to Station A, indicating that more EVs prefer to choose Station B as their optimal strategy. The reason for this can be attributed to the fact that when the waiting time at Station A increases, EV drivers who would have chosen Station A may opt for Station B instead, as the total time cost at Station A has increased. In the extreme case where $w_{A}^{x}$ is infinitely large, the optimal strategies for all EVs are to charge at Station B since they will never start charging at Station A. 
	
	Furthermore, we lower the unit distance traveling time $\tau$ to $0.5$ while keeping the other parameters the same as in the first case. As a result, the indifference curve shifts towards Station A. This can be explained by the fact that the importance of traveling time cost diminishes when the unit distance traveling time decreases. Therefore, some EV drivers who are restricted by the large traveling time and decide to choose Station A are no longer concerned about that and may choose Station B instead.
	
	Finally, we investigate the impact of power consumption over unit distance, $c$, by decreasing it to $0.1$ while keeping the other parameters the same as in the first case. As shown in Fig. \ref{different parameters}, the indifference curve slightly shifts downwards towards Station A. Thus, more EVs tend to choose Station B for charging. The reason is as follows. When $c$ is small, the power consumed by EVs during their travels to charging stations is also reduced. As a result, the remaining battery levels upon arrival at two charging stations become almost the same, regardless of their initial positions. For EVs on the indifference curve in the first case, they are indifferent to both choices even if the charging time at Station B is larger than that at Station A. Therefore, if the charging time difference is reduced or even eliminated, these previously indifferent drivers will choose Station B instead as their optimal strategy, causing the indifference curve to shift downwards when $c$ decreases.
	
	
    
	
	
    \subsection{{Field data test}}
    {Using the field data and charging rate curves shown in Fig. \ref{field data charging rates}, we plot the indifference curves in Fig. \ref{field data indifference curves} with the parameters of $c = 0.2$, $\tau = 1$, $r_{t} = 1$, $w_{A}^{x} = 0.5$ and $w_{B}^{x} = 0$.} As evidenced by the real-world data, the charging rates decrease as the charge level increases. Generally speaking, with lower remaining battery levels, more EV drivers should opt for a farther charging station in order to achieve a lower overall time cost. This finding is interesting because EV drivers with low battery levels may have charging anxiety. So psychologically, they would have a stronger willingness to opt for a closer charging station than those people with high battery levels. However, our findings show that things should be the opposite if they are rational people. Besides of that, the results also suggest that charging station selection of EV owners will be significantly influenced by their car's charging curve characteristics.
	
{

    \section{Numerical Case Study for Endogenous Model}\label{numerical endogenous model}
        
        In this section, we focus on endogenous model and compare our findings with those obtained from the exogenous model. First, we explore cases with explicit charging rate functions and present the optimal strategies for EV drivers. Here the charging rates are considered to be flat, decreasing and piecewise flat. Subsequently, we conduct a field data test based on the charging rates shown in Fig. \ref{field data charging rates}.

        \subsection{Flat charging rates}
        We first investigate the case of flat charging rates, where the charging time function is set as $F(r)=1-r$. This is the simplest case, as the charging speed is constant, regardless of the remaining battery levels. With the parameters set as $w_A^{x}=1$, $w_B^{x}=0$, $\tau=1$, $c=0.2$, $r_t=1$ and $\epsilon=1$, Fig. \ref{fig: flat charging rate} compares the indifference curve obtained using exogenous and endogenous models, denoted by $g^x$ and $g^n$, respectively. We observe that there is a large gap between two indifference curves, due to the presence of congestion in endogenous model.

        \begin{figure}[ht]
            \subfigure[Flat charging rates]{
                \begin{minipage}[t]{0.5\linewidth}
    			\centering
    			\includegraphics[width=\linewidth]{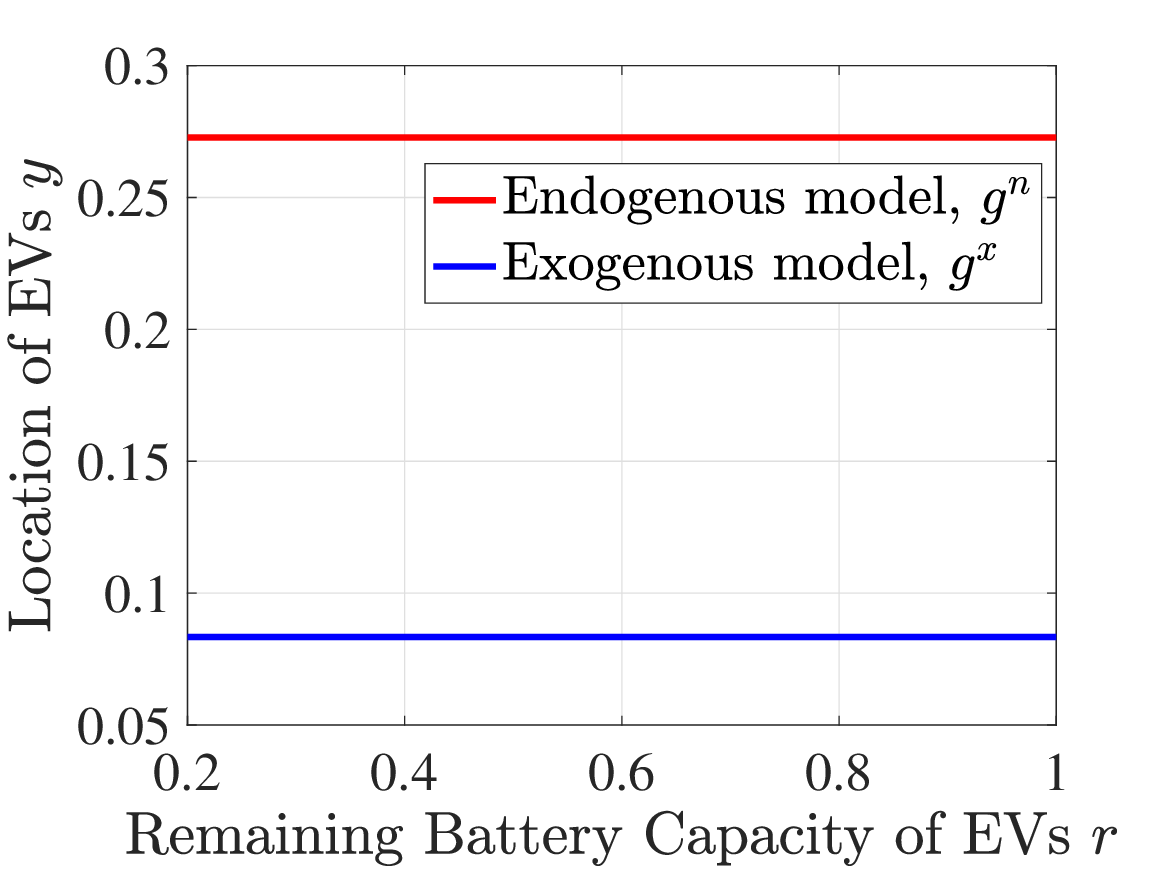}
    			\label{fig: flat charging rate}
    		\end{minipage}
            }%
            \subfigure[Decreasing charging rates]{
                \begin{minipage}[t]{0.5\linewidth}
    			\centering
    			\includegraphics[width=\linewidth]{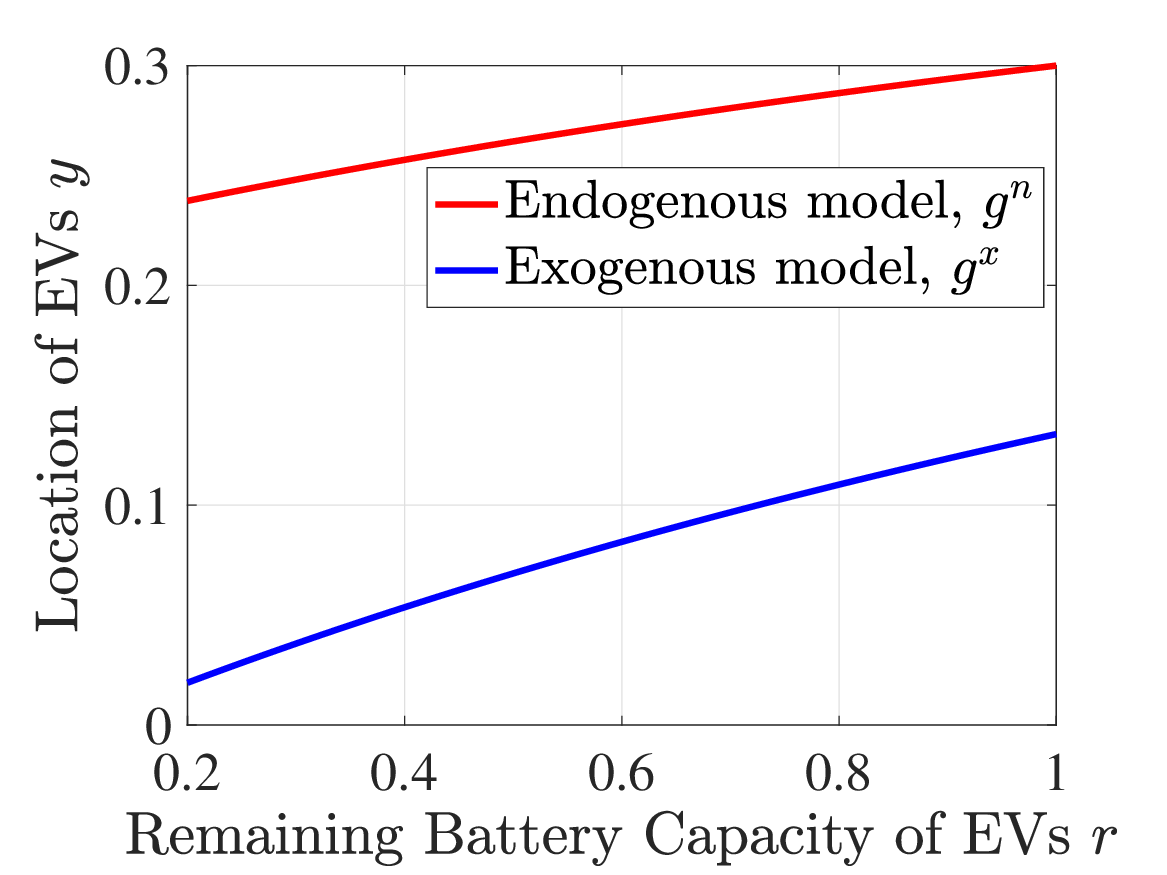}
    			\label{fig: decreasing charging rate}
    		\end{minipage}
            }
            \caption{Comparison of the indifference curves in exogenous and endogenous models.}
        \end{figure}

        \subsection{Decreasing charging rates}
        Next, we consider the decreasing charging rates scenario and assume the charging time function to be $F(r)=1-r^2$. With the same parameters, the indifference curve is shown and compared in Fig. \ref{fig: decreasing charging rate}. Similar to the first case, as the congestion at Station B is higher than A, more EV drivers choose Station A instead to avoid the congestion.

        \subsection{Piecewise flat charging rates}
        In this subsection, we study a more practical case where the charging rates are piecewise flat. It corresponds to the scenario where the charging rate significantly drops after the battery level exceeds a certain threshold. An example of such charging rate characteristic is Ford Mustang Mach-e, where the charging speed decreases once the battery level reaches $80\%$. For simplicity, we let the charging rate curves be piecewise flat, as shown in Fig. \ref{piecewise linear P}. With the same parameters, the resulting indifference curve is demonstrated in red in Fig. \ref{piecewise linear}. Meanwhile, the curve in blue is the indifference curve of exogenous model. 

        \begin{figure}[ht]
            \subfigure[Charging rate curves]{
                \begin{minipage}[t]{0.5\linewidth}
    			\centering
    			\includegraphics[width=\linewidth]{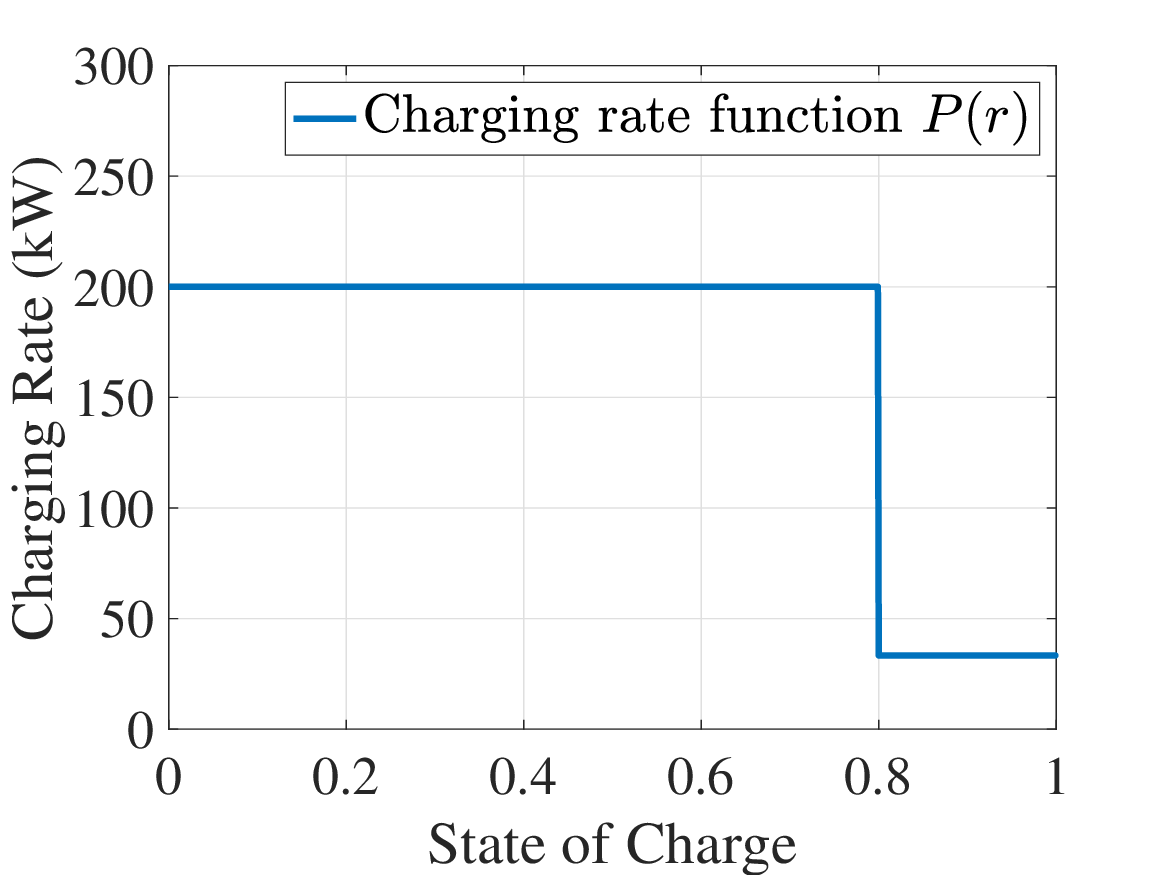}
    			\label{piecewise linear P}
    		\end{minipage}
            }%
            \subfigure[Indifference curves]{
                \begin{minipage}[t]{0.5\linewidth}
    			\centering
    			\includegraphics[width=\linewidth]{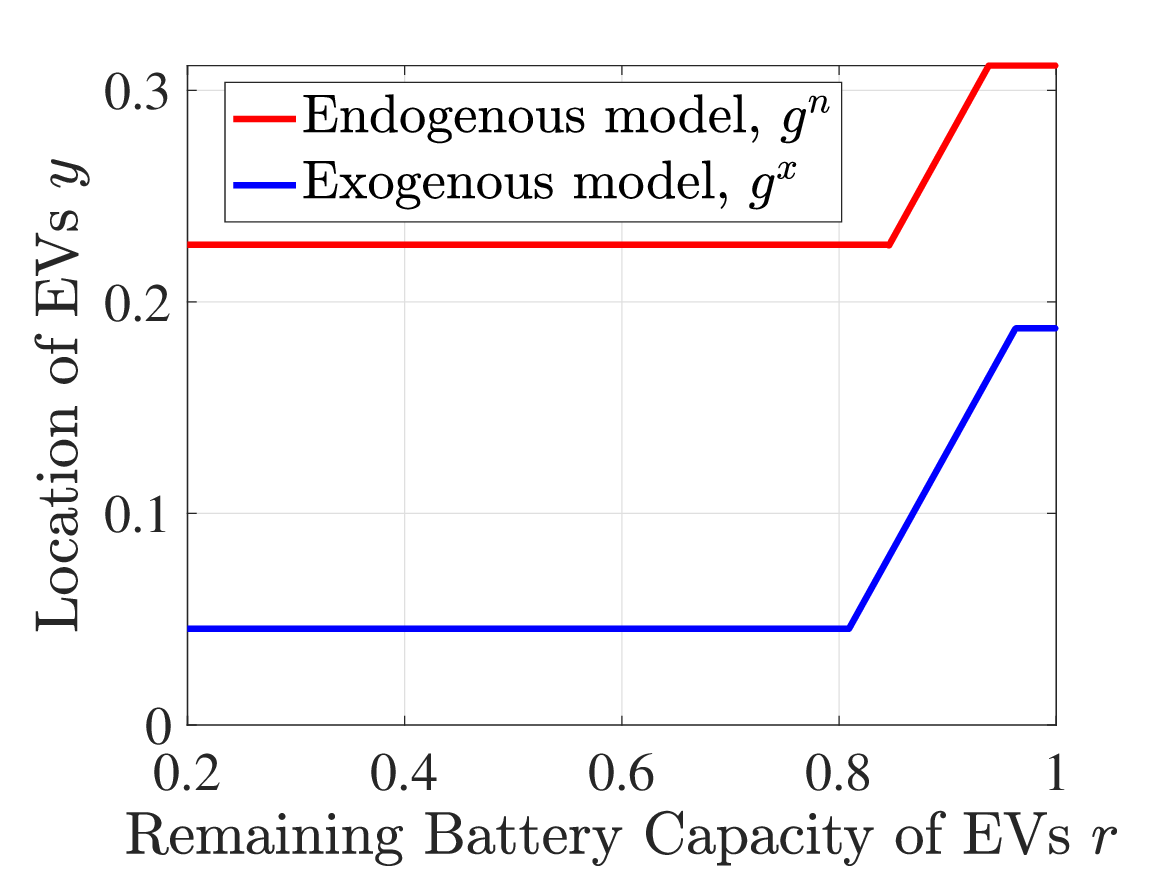}
    			\label{piecewise linear}
    		\end{minipage}
            }
            \caption{Charging rates and indifference curves in piecewise flat charging rates scenario.}
        \end{figure}

        Combining all scenarios together, we see that the inclusion of endogenous waiting times in the model reveals a significant difference in the indifference curves and the optimal strategies, compared to the exogenous results. More EVs choose not to wait at a more ``crowded'' charging station. Many of those who previously opt for Station B now switch to Station A. Although they suffer from a longer exogenous waiting time at Station A, the congestion time they experience is much reduced since only a few people may be waiting at that charging station. This shows the impact of considering endogenous model on EVs' charging station choices.

        \subsection{Field data test}
        Finally, we present the results of testing on real-world data. To perform this experiment, the parameters are set as $c = 0.2$, $\tau = 1$, $r_{t} = 1$, $w_{A}^{x} = 0.5$, $w_{B}^{x} = 0$ and $\epsilon=1$. The charging rates used in this experiment are the same as those shown in Fig. \ref{field data charging rates}. The solutions $g(r,z)$, or say indifference curves, are plotted in Fig. \ref{field data endogenous}, and the results are compared to those obtained using the exogenous model. 
        
        We find that many EVs change their optimal strategies from Station B to Station A since the congestion at Station B is larger than that at Station A. The values of $\alpha_\text{Audi}$, $\alpha_\text{Ford}$ and $\alpha_\text{Tesla}$, which are the proportions of EVs of each make and model that choose Station A as their optimal strategies, are calculated to be 0.3814, 0.3876 and 0.3837, respectively.


        \begin{figure}[ht]
        \begin{minipage}[b]{0.48\linewidth}
        \centering
        \includegraphics[width=\textwidth]{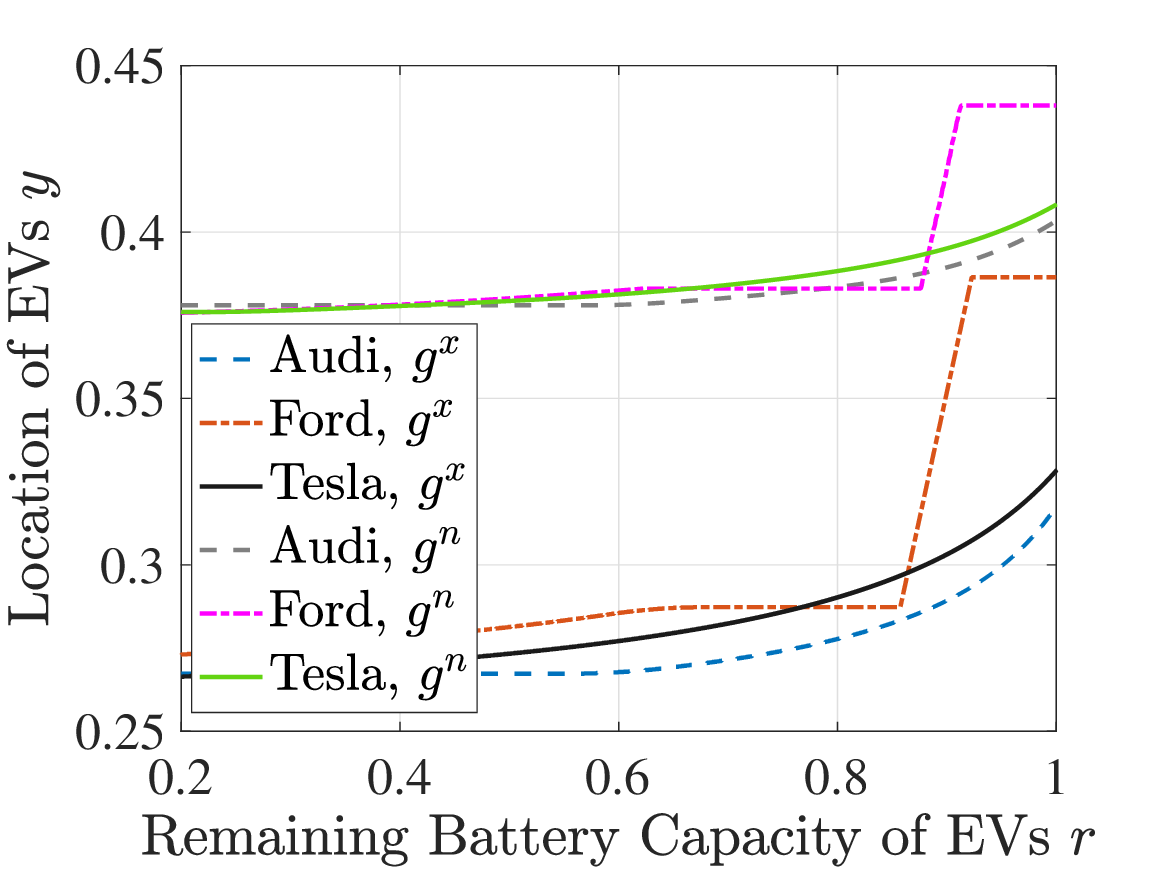}
        \caption{Field data test on endogenous waiting time model.}
        \label{field data endogenous}
        \end{minipage}
        \hspace{0.1cm}
        \begin{minipage}[b]{0.48\linewidth}
        \centering
        \includegraphics[width=\textwidth]{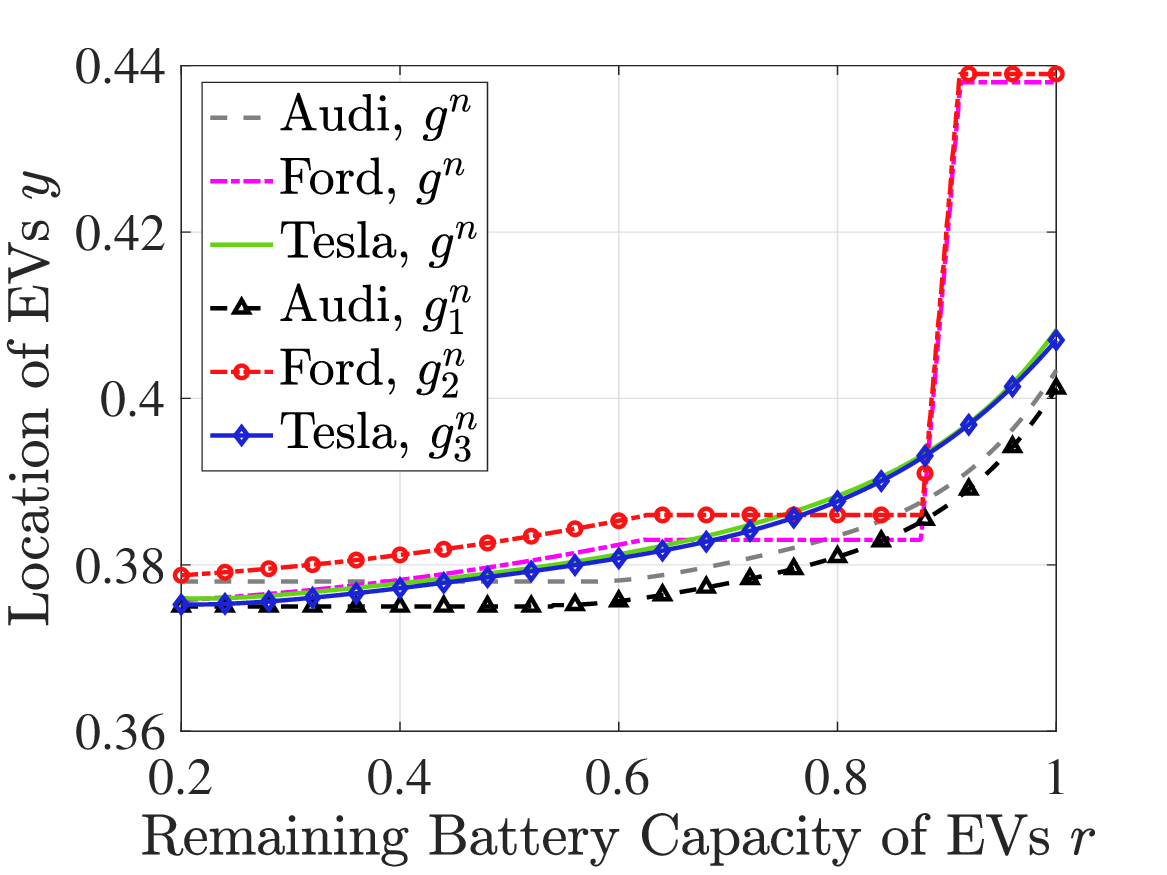}
        \caption{Field data test on heterogeneous model.}
        \label{field data heterogeneous}
        \end{minipage}
        \end{figure}

\section{Case Study on Heterogeneous Model}\label{heterogeneous}
    In this section, we discuss the heterogeneous model, where different EVs may possess distinct charging rate characteristics. From the previous sections, we have seen that different charging rate characteristics lead to different indifference curves and thus different optimal strategies. Therefore, we are interested in the case when EVs are mixed, for example, one half drives Tesla and the other half drives Audi, and how their strategies may be affected by each other.
    
    To effectively model the behavior of these diverse EVs, we classify them into multiple classes, say $m$ classes, based on their charging rate characteristics. Denote the charging rate and time functions by $P_1,\dots,P_m$ and $F_1,\dots,F_m$, respectively. Additionally, we assign a normalized weight to each class, denoted by $W_1,\dots,W_m$, reflecting the number of EVs in each class. Note that all weights add up to 1, i.e., $\sum_{i=1}^m W_i=1$. 
    


    In exogenous waiting time model, each class of EVs has its own initial value problem and associated solution, which means that the optimal strategies for each class of EVs are not affected by other classes of users. 

    However, in endogenous waiting time model, the solutions to the initial value problem and optimal strategies for each class of EVs are coupled with the other classes. This means that the behavior of one class affects the behavior of other classes, leading to changes in their charging strategies and resulting in a more complex decision making process. Specifically, when considering the endogenous model with multiple classes of EVs, the interaction between the classes can lead to new patterns of behavior. For example, some classes may prioritize charging at one station over the other one, leading to high congestion. Mathematically, the problem can be formulated as solving a system of first order ordinary differential equations with an initial value that combines the weighted solutions of each ODE. {Note that our theory only works for $m=1$, but the approach is general and can be extended to the case of $m>1$. We leave the technical details as future work and present an experimental result here.}

        
    
{In $m>1$ setup, $A$ is defined as the weighted sum of  integral of ODE solutions $g_1, \dots, g_m$ intersecting region $R$. With appropriate initial conditions, $A$ coincides with the congestion $\alpha$. Numerically, we solve it by iteratively altering the initial conditions so that the difference between $A$ and $\alpha$ is minimized.} To demonstrate how the heterogeneous model differs from the single-class homogeneous model, we use the same field data as before, with the parameters being $c = 0.2$, $\tau = 1$, $r_{t} = 1$, $w_{A}^{x} = 0.5$, $w_{B}^{x} = 0$ and $\epsilon=1$. In this case, we assume that the EVs consist of Audi, Ford and Tesla users, with the same weight for each class, i.e., $W_1=W_2=W_3=1/3$.
    
    The charging rate functions are shown in Fig. \ref{field data charging rates}. The indifference curves are illustrated in Fig. \ref{field data heterogeneous}, where $g_1^{n}$, $g_2^{n}$ and $g_3^{n}$ represent the indifference curves for each class. We also compare these results with the single-class case using the endogenous waiting time model, as shown by $g^{n}$. In this experiment, for each class we calculate
    \begin{align*}
        \beta_i = \frac{\text{The number of EVs in class $i$ that choose A}}{\text{The total number of EVs in class $i$}},
    \end{align*}
    and obtain $\beta_1=0.3787$, $\beta_2=0.3903$ and $\beta_3=0.3832$.
        
    After accounting for the mutual effects between different classes of EVs, we observe that more Audi users opt for Station B while more Ford users opt for Station A. This is because Audi users generally have faster charging rates, making their charging time difference at two stations relatively small. Therefore, they are more likely to travel a bit further to tradeoff their charging time and waiting time. As a consequence, the indifference curve shifts downwards. In contrast, Ford users have slower charging rates and are more likely to choose a closer station. Moreover, by taking Audi users as an example, we find that the total time cost of indifferent drivers in heterogeneous case (plotted as Audi, $g_1^n$) is slightly increased compared to homogeneous case (plotted as Audi, $g^n$). The reason is due to the increase in congestion level at Station A. On the other hand, the indifferent drivers in Ford class, who have slower charging rates, get a slightly reduced total time cost when Audi and Tesla users participate. Therefore, as the EV users with the slowest charging rates, they actually benefit from those EVs with faster charging rates. This highlights the externalities EV drivers impose on each other.


}

\section{Conclusion}\label{conclusion}
{
    We study the charging station selection problem  with the consideration of varying charging rates. Our work introduces an ODE approach to derive the optimal strategies for EV drivers considering both exogenous and endogenous models. In the exogenous model, each EV has a fixed, predetermined waiting time that is not influenced by the behavior of other EVs. On the other hand, the endogenous model takes into account the impact of other EVs, where an EV's waiting time is affected by the charging decisions of others. We find that even if EVs have the same location, their optimal strategies might be opposite depending on the remaining battery capacities. The EVs in the region $R$ could be divided into two groups choosing different charging stations. {We use a unified ODE framework to solve the extended indifference curve for both exogenous and endogenous models, and then give the optimal strategies for EVs.} Finally, we demonstrate the optimal strategies by numerical simulations and field data tests. 

    Overall, our study provides a practical solution for EV users to optimally make decisions on charging station selection. In future work, we plan to further refine our model by incorporating more complex user behaviors, such as irrational behaviors, to better capture real-world charging scenarios.
    
    }



\ifCLASSOPTIONcaptionsoff
  \newpage
\fi


\begin{thebibliography}{1}

\bibitem{2030 goal} The ROUTES Initiative, ``Electric vehicles \& rural transportation,'' \textit{U.S. Department of Transportation}. [Online]. Available: https://www.transportation.gov/rural/ev. [Accessed: 08-Aug-2023].
    \bibitem{EV parking conference} B. Badia, R. Berry and E. Wei, ``Price and capacity competition for EV parking with government mandates,'' \textit{56th Annual Allerton Conference on Communication, Control, and Computing (Allerton)}, pp. 583-589, 2018.
    \bibitem{EV parking journal} B. Badia, R. A. Berry and E. Wei, ``Investment in EV charging spots for parking,'' \textit{IEEE Transactions on Network Science and Engineering}, vol. 7, no. 2, pp. 650-661, 2020.
    \bibitem{battery swapping} P. You, J. Z. F. Pang and S. H. Low, ``Online station assignment for electric vehicle battery swapping,'' \textit{IEEE Transactions on Intelligent Transportation Systems}, vol. 23, no. 4, pp. 3256-3267, April 2022.
    \bibitem{BMS} K. W. E. Cheng, B. P. Divakar, H. Wu, K. Ding and H. F. Ho, ``Battery-management system (BMS) and SOC development for electrical vehicles,'' \textit{IEEE Transactions on Vehicular Technology}, vol. 60, no. 1, pp. 76-88, Jan. 2011.
    \bibitem{battery behavior} C. Sun, T. Li, S. H. Low and V. O. K. Li, ``Classification of electric vehicle charging time series with selective clustering,'' \textit{Electric Power Systems Research}, vol. 189, pp. 106695, 2020.
    \bibitem{battery model} A. Seaman, TS. Dao and J. McPhee, ``A survey of mathematics-based equivalent-circuit and electrochemical battery models for hybrid and electric vehicle simulation,'' \textit{Journal of Power Sources}, vol. 256, pp. 410–423, 2014. 
    \bibitem{charging rate at 90 percent} The ROUTES Initiative, ``Electric vehicle charging speeds,'' \textit{U.S. Department of Transportation}. [Online]. Available: https://www.transportation.gov/rural/ev/toolkit/ev-basics/charging-speeds. [Accessed: 08-Aug-2023].
    \bibitem{EV charging coordination} P. Alexeenko and E. Bitar, ``Achieving reliable coordination of residential plug-in electric vehicle charging: A pilot study,'' \textit{arXiv preprint arXiv:2112.04559v2}, 2022.
    \bibitem{Ridesharing system} T. Mamalis, S. Bose and L. R. Varshney, ``Ridesharing systems with electric vehicles,'' \textit{2019 American Control Conference (ACC)}, pp. 3329-3334, 2019.
    \bibitem{scheduling algorithm for ACN} B. Alinia, M. H. Hajiesmaili, Z. J. Lee, N. Crespi and E. Mallada, ``Online EV scheduling algorithms for adaptive charging networks with global peak constraints,'' \textit{IEEE Transactions on Sustainable Computing}, vol. 7, no. 3, pp. 537-548, 1 July-Sept. 2022.
    \bibitem{eco-driving control} Y. Zhang, X. Qu and L. Tong, ``Optimal eco-driving control of autonomous and electric trucks in adaptation to highway topography: Energy minimization and battery life extension,'' \textit{IEEE Transactions on Transportation Electrification}, vol. 8, no. 2, pp. 2149-2163, June 2022.
    \bibitem{V2G} J. Qin, K. Poolla and P. Varaiya, ``Mobile storage for demand charge reduction,'' \textit{IEEE Transactions on Intelligent Transportation Systems}, vol. 23, no. 7, pp. 7952-7962, July 2022.
    \bibitem{optimization model} A. Y. S. Lam, Y. -W. Leung and X. Chu, ``Electric vehicle charging station placement: formulation, complexity, and solutions,'' \textit{IEEE Transactions on Smart Grid}, vol. 5, no. 6, pp. 2846-2856, Nov. 2014.
    \bibitem{model with EV bus} X. Wang, C. Yuen, N. U. Hassan, N. An and W. Wu, ``Electric vehicle charging station placement for urban public bus systems,'' \textit{IEEE Transactions on Intelligent Transportation Systems}, vol. 18, no. 1, pp. 128-139, Jan. 2017.
    \bibitem{model with self-interested EVs} Y. Xiong, J. Gan, B. An, C. Miao and A. L. C. Bazzan, ``Optimal electric vehicle fast charging station placement based on game theoretical framework,'' \textit{IEEE Transactions on Intelligent Transportation Systems}, vol. 19, no. 8, pp. 2493-2504, Aug. 2018.
    \bibitem{model with various networks} C. Luo, Y. Huang and V. Gupta, ``Placement of EV charging stations—balancing benefits among multiple entities,'' \textit{IEEE Transactions on Smart Grid}, vol. 8, no. 2, pp. 759-768, March 2017.
    \bibitem{tesla model 3} Fastned, ``Tesla,'' \textit{Fastned FAQ}. [Online]. Available: https://support.fastned.nl/hc/en-gb/articles/360012178313-Tesla. [Accessed: 08-Aug-2023].
    \bibitem{field data} B. Nyland, ``Charging curves,'' \textit{Google Sheets}. [Online]. Available: https://drive.google.com/drive/folders/\\1HOwktdiZmm40atGPwymzrxErMi1ZrKPP. [Accessed: 08-Aug-2023].
    \bibitem{hotelling model} H. Hotelling, ``Stability in competition,'' \textit{The Economic Journal}, vol. 39, no. 153, pp. 41–57, 1929.
    \bibitem{allerton} L. Yi and E. Wei, ``The effects of varying charging rates on optimal charging station choices for electric vehicles,'' \textit{58th Annual Allerton Conference on Communication, Control, and Computing (Allerton)}, Monticello, IL, USA, 2022, pp. 1-5.
    \bibitem{linear congestion level} R. Arnott and E. Inci, ``An integrated model of downtown parking and traffic congestion,'' \textit{J. Urban Econ.}, vol. 60, no. 3, pp. 418-442, 2006.
    \bibitem{rudin} W. Rudin, \textit{Principles of Mathematical Analysis,} 3rd ed., New York, USA: McGraw-Hill, 1976, ch. 6, pp. 133.
    \bibitem{Picard} W. E. Boyce, R. C. DiPrima, and D. B. Meade, \textit{Elementary differential equations and boundary value problems,} 11th ed., Hoboken, NJ, USA: Wiley, 2017, ch. 2, pp. 51–53. 
    \bibitem{arnold} V. I. Arnol'd, \textit{Ordinary Differential Equations,} 3rd ed., New York, USA: Springer-Verlag, 1992, ch. 4, pp. 276.
\end{thebibliography}
\end{document}